\documentclass[preprint]{elsarticle}
\usepackage{amssymb,amsmath,amsfonts,amsthm,mathtools}

\def\Z{\mathbb Z}
\def\N{\mathbb N}
\def\C{\mathbb C}
\def\R{\mathbb R}
\def\Q{\mathbb Q}

\def\F{\mathbb F}
\def\pfz{\begin{proof}}
\def\pfk{\end{proof}}
\renewcommand{\S}{\Sigma}
\newcommand{\Om}{\Omega}
\def\L{\mathcal L}
\newcommand{\al}{\sim_{\mathtt{alg}}}

\renewcommand{\L }{\mathcal{L}}
\newcommand{\K }{\mathcal{K}}

\renewcommand{\bf}[1]{\mathbf{#1}}
\newcommand{\bs}[1]{\boldsymbol{#1}}
\renewcommand{\Re}{\mathtt{Re}}
\renewcommand{\Im}{\mathtt{Im}}
\newcommand{\m}[1]{\mu\, _{\Q, #1}}

\usepackage{tikz}      
\usetikzlibrary{chains}
\usetikzlibrary{shapes.geometric}
\usetikzlibrary{quotes,angles,positioning}
\usepackage{pdfpages}

\newtheorem{lem}{Lemma}[section]
\newtheorem{thm}[lem]{Theorem}
\newtheorem{prop}[lem]{Proposition}
\newtheorem{coro}[lem]{Corollary}
\newtheorem{definition}[lem]{Definition}
\newtheorem{example}[lem]{Example}
\newtheorem{pozn}[lem]{Remark}
\newtheorem{claim}[lem]{Claim}
\newtheorem{que}{Question}

\begin{document}

\begin{frontmatter}

\title{On generalized self-similarities of cut-and-project sets}


\author[fjfi]{Zuzana Mas\'akov\'a}
\ead{zuzana.masakova@fjfi.cvut.cz}

\author[fjfi]{Jan Maz\'a\v c\corref{cor1}}
\ead{hanis.mazac@gmail.com}
\cortext[cor1]{Corresponding author}

\author[fjfi]{Edita Pelantov\'a}
\ead{edita.pelantova@fjfi.cvut.cz}

\address[fjfi]{FNSPE, Czech Technical University in Prague, Trojanova 13, 120 00 Praha 2, Czech~Republic}

\begin{abstract}
Cut-and-project sets $\Sigma\subset\R^n$ represent one of the types of uniformly discrete relatively dense sets. They arise by projection of a section of a higher-dimensional lattice to a suitably oriented subspace. Cut-and-project sets find application in solid state physics as mathematical models of atomic positions in quasicrystals, the description of their symmetries is therefore of high importance. We focus on the question when a linear map $A$ on $\R^n$ is a self-similarity of a cut-and-project set $\Sigma$, i.e.\ satisfies $A\Sigma\subset\Sigma$. We characterize such mappings $A$ and provide a construction of a suitable cut-and-project set $\Sigma$. We determine minimal dimension of a lattice which permits construction of such a set $\Sigma$.
\end{abstract}


\end{frontmatter}

\section{Introduction}

Although the principle of cut and projection was used in some sense already in the twenties in the theory of quasiperiodic functions~\cite{Bohr}, the main attention the cut-and-project method gained after the discovery of non-crystallographic materials with long range order -- the so-called quasicrystals. For an overview of the contemporary knowledge on mathematical modelling of quasicrystals, see~\cite{Aperiodic1}. Cut-and-project sets nowadays appear also in connection to non-standard numeration systems~\cite{BuFrGaKr} or symbolic dynamical systems and combinatorics on words, see e.g~\cite{LunnonPleasants}. Cut-and-project sets are recognized as suitable mathematical models of quasicrystals mainly because of two important properties: they belong to the family of Delone sets of finite type and have no translational symmetry. Recall that models of crystals are based on the notion of a lattice.
The cut-and-project method consists in considering a lattice $\L \subset \R^s$ and projections $\pi_\parallel$, $\pi_\perp$ to two orthogonal subspaces of $\R^s$, called physical and internal spaces, respectively, say of dimensions $n$, $s-n$. Taking only the $\pi_\parallel$-image of those lattice points whose $\pi_\perp$-image fits into a~chosen bounded set $\Om$, the so-called window, one gets the cut-and-project set $\Sigma(\Om)\subset \R^n$. Under certain assumptions on the projections and the window $\Omega$, the set $\Sigma(\Omega)$ displays the required properties.

The first experimentally discovered quasicrystal was a manganese-alu\-minium alloy whose dif\-fraction pattern revealed 10-fold symmetry~\cite{Shechtman}. The well known crystallographic restriction~\cite{Aperiodic1} however implies that a discrete periodic structure in dimension 2 or 3 with such symmetry cannot exist. It was recognized in~\cite{KramerNeri} that for a model of such a material one can take an aperiodic Delone set obtained by projection of a 4-dimensional lattice. For, dimension 4 is the smallest containing a lattice invariant under an isometry of order 10. A similar construction was then provided by Niizeki~\cite{Niizeki} in the algebraic context of cyclotomic fields, Moody and Patera~\cite{icosians} with the help of quaternions, and Barache et al.~\cite{Barache} using root lattices of finite Coxeter groups. For a cut-and-project set with a symmetry of order $r$, one needs to start with a lattice $\L$ having such a symmetry. It was derived in~\cite{Hiller} that minimal dimension $s$ in which such a lattice exists is given by $\phi_a(r)$ where $\phi_a$ is the additive version of the Euler totient function $\phi$.
A planar cut-and-project set revealing symmetry of order $r$ must have been constructed by projection of a lattice $\L$ in dimension at least $\phi(r)$, see~\cite{BJS}.

Construction of a cut-and-project set with predefined symmetries (isometries $R$ such that the discrete set $\Sigma$ satisfies $R\Sigma=\Sigma$) is also the subject of~\cite{Pleas}. For a given finite isometry group $G$, Pleasants shows that there is a cut-and-project set $\Sigma(\Omega)$ such that every $g\in G$ is a symmetry of $\Sigma(\Omega)$. Similar objective is in focus of Cotfas~\cite{Cotfas-clusters} who constructs a cut-and-project set as a quasi-periodic packing of interpenetrating copies of a $G$-cluster, i.e.\ a union of orbits of the finite symmetry group $G$.

Besides symmetries, quasicrystal models usually reveal all kinds of self-similarities, i.e.\ affine mappings $A$, under which the model is closed, $A\Sigma\subset\Sigma$. Such self-similarities have been studied, in particular, in the case when the affine mapping is just a scaling~\cite{BermanMoody,inflcenters,Cotfas-Gmodelsetselfsimilarities}. Lagarias~\cite{Lagarias} gave conditions on scaling factors of Delone sets of finite type which include cut-and-project sets. Pleasants~\cite{Pleas} studies a slightly different concept, the so-called flation property.

For us, a self-similarity of a cut-and-project set is any linear  map $A$ of $\Sigma(\Omega)$ into itself.
We mainly focus on the following question:

\begin{que}\label{q1}
For a given linear mapping $A$ decide if there exists a suitable cut-and-project set $\Sigma(\Omega)$ such that $A\Sigma(\Omega) \subset \Sigma(\Omega)$ and if so, provide a construction.
\end{que}

Answering this question naturally implies finding suitable lattice $\L\subset \R^s$ and projections $\pi_\parallel$ and $\pi_\perp$ such that $A\pi_\parallel(\L) \subset \pi_\parallel(\L)$
and there is a linear action also in the internal space, cf.\ Question~\ref{q2}. We say that such $A$ is a self-similarity of the cut-and-project scheme. We prove that for a linear mapping $A$ diagonalisable over $\C$ there exists a~generic cut-and-project scheme $\Lambda=(\L\subset\R^s,\R^n)$ with self-similarity
$A$ if and only if the spectrum of $A$ is composed of algebraic integers (Theorem~\ref{t:mindimensionnotmin}). We provide a construction yielding minimal dimension $s$ possible (Theorem~\ref{t:mindimension}).
By this, we extend the crystallographic restriction determining minimal dimension for obtaining $r$-fold symmetry in a planar quasicrystal.

Consequently, we formulate an exhaustive answer to Question~\ref{q1} for non-singular mappings $A$ diagonalizable over $\C$, see Theorem \ref{thm:npp_for_cps}. We also determine the minimal dimension of a lattice $\L$ allowing construction of a cut-and-project set with self-similarity $A$. This generalizes the known results by Lagarias~\cite{Lagarias}, see Section~\ref{sec:comments}.

\section{Preliminaries}

This section aims to recall necessary basic notions from the algebraic number theory and linear algebra. Readers being familiar with these parts of mathematics may freely skip this section.

A complex number $\eta$ is an algebraic number if there exists a monic polynomial $f \in \Q[X]$ such that $f(\eta) = 0$.
If $f$ is of minimal degree, it is called the minimal polynomial of $\eta$ and its degree is the degree of $\eta$. The other roots of the minimal polynomial $f$ are algebraic conjugates of $\eta$.
If $\nu, \eta$ are algebraic conjugates, we write $\nu\al\eta$.
If $\eta$ is a root of a monic polynomial with integer coefficients, then it is an algebraic integer. 

%
%

For a set of complex numbers $t_1,\dots,t_m$ we denote by $\Q(t_1,\dots,t_m)$ the minimal subfield of $\C$ containing $t_1,\dots,t_m$. If $t_j$ are algebraic for every $j$, it is known that there exists an algebraic number $\alpha$ such that $\Q(t_1,\dots,t_m)=\Q(\alpha)$. If $\alpha$ is such an algebraic number of degree $d$, $\Q(\alpha)$ is said to be an algebraic number field or algebraic field extension of $\Q$ by $\alpha$ of degree $d$. The degree $d$ of $\Q(\alpha)$ is the dimension of $\Q(\alpha)$ as a vector space over $\Q$, thus
$$
\Q(\alpha)=\{c_0 + c_1 \alpha + \dots + c_{d-1}\alpha ^{d-1} : c_0,\dots,c_{d-1}\in\Q\}.
$$
If $\beta\al\alpha$, then the mapping $\psi : \Q(\alpha) \to \Q(\beta) $ defined by
	\[
\psi(c_0 + c_1 \alpha + \dots + c_{d-1}\alpha ^{d-1}) = c_0 + c_1 \beta + \dots + c_{d-1}\beta ^{d-1}
   \]	
is a field isomorphism.

Let $f\in \Q[X]$ be a polynomial. The field extension $\F$ of $\Q$ is the splitting field of the polynomial $f$ if $\F$ contains all roots of the polynomial $f$.
Given a polynomial $f(X) = X^d - \sum _{i=0}^{d-1} a_i X^i \in \Q[X]$ we often make use of its companion matrix $C_f \in \Q^{d\times d}$, namely
$$
C_f=\begin{pmatrix}
0&1&\cdots&0\\[-1mm]
\vdots&\vdots&\ddots&\vdots\\
0&0&\cdots&1\\
a_0&a_1&\cdots&a_{d\!-\!1}
\end{pmatrix}\,.
$$

In this paper we heavily use the matrix formalism for our study. Recall the Kronecker (or tensor) product $A\otimes B$ of matrices $A\in \C^{n \times m}$ and $B \in \C^{p\times q}$ defined as a matrix of dimension $np \times mq$
\[
A\otimes B = \begin{pmatrix}
a_{11} B & \cdots & a_{1m}B \\
\vdots & \ddots & \vdots \\
a_{n1} B & \cdots & a_{nm}B
\end{pmatrix}.
\]

Recall (see e.g.~\cite{Fiedler}) that every complex matrix is similar to a matrix in Jordan normal form. If the matrix has real entries, we can modify it into real Jordan form. In particular, for every $T\in\R^{n\times n}$, there exists non-singular $W\in\R^{n\times n}$ such that $W^{-1}TW= {\bigoplus_{k}J_k}$ where $J_k$ is a real Jordan block. A real Jordan block corresponding to the eigenvalue $\lambda$ of $T$ is of the form
$$
J(\lambda)=\begin{pmatrix}
R(\lambda)&I&& \\
&R(\lambda)&\ddots& \\
&&\ddots&I \\
&&&R(\lambda)
\end{pmatrix}\ \text{with}\
R(\lambda)=\begin{cases}
\ (\lambda) & \mbox{if } \lambda\in\R, \\[2mm]
\left(\begin{smallmatrix}
\Re\lambda & \Im\lambda \\[1mm]
-\Im\lambda & \Re\lambda
\end{smallmatrix}\right) & \mbox{if } \lambda\in\C\setminus\R.
\end{cases}
$$
Note that $I$ is a unit matrix of order 1 or 2, according to the order of $R(\lambda)$.
In cases that the real matrix $T$ is diagonalizable over $\C$, the real Jordan form  of $T$ reduces to the classical quasidiagonal form.

For rational matrices $C$ several different rational forms are used. We use a rational Jordan form  which reflects decomposition into the maximal number of cyclic subspaces (spanned by some vector and its repeated images under $C$) over $\Q$ of the matrix $C$. In particular, for every $C\in\Q^{s\times s}$ there exists a non-singular matrix $W\in\Q^{s\times s}$ such that
$W^{-1}CW= {\bigoplus_{k}J_k}$ where the rational Jordan blocks correspond to polynomials $f$ irreducible over $\Q$ which are factors of the characteristic polynomial of $C$. The blocks are of the form
$$
J=\begin{pmatrix}
C_f&I&& \\
&C_f&\ddots& \\
&&\ddots&I \\
&&&C_f
\end{pmatrix}
$$
where the block $C_f$ denotes the companion matrix of $f$.
The unit matrix $I$ in the Jordan block $J$ is of order $d$.

Given a matrix $T\in\C^{d\times d}$, we denote its spectrum by $\sigma(T)$ and its spectral radius by $\varrho(T)$. For any complex number $\lambda$ we denote by $m_T(\lambda)$ its algebraic multiplicity as an eigenvalue of $T$. Note that $m_T(\lambda)=0$ if $\lambda\notin\sigma(T)$.

For a matrix $T\in\C^{d\times d}$ we define its minimal polynomial $\mu_T$ as the monic polynomial in $\C[X]$ of the smallest degree such that $\mu_T(T)$ is the zero matrix. According to the Hamilton-Cayley theorem, the minimal polynomial
of $T$ is of degree at most $d$ and its roots are the eigenvalues of $T$. The polynomial $\mu_T$ can be calculated using the Jordan canonical form of $T$, see~\cite{Fiedler}.
It is easily seen that every polynomial which is annihilated by $T$ is divisible by $\mu_T$.

If the spectrum of the matrix $T$ is composed of algebraic numbers we also define the minimal polynomial $\m{T}$ of $T$ over $\Q$ as the monic polynomial in $\Q[X]$ of the smallest degree such that $\m{T}(T)=O$. The degree of $\m{T}$ may be greater than $d$. From the construction of the minimal polynomial using the Smith normal form it follows that for a matrix $C\in\Q^{d\times d}$, we have $\mu_C=\m{C}$. The following properties can be easily shown:
\begin{itemize}
\item every eigenvalue of $T$ is a root of $\m{T}$,
\item every root of $\m{T}$ is an algebraic conjugate of an eigenvalue of $T$, and
\item if $T$ is diagonalizable over $\C$, then $\m{T}$ is a product of distinct monic polynomials in $\Q[X]$ irreducible over $\Q$.
\end{itemize}

\section{Cut-and-project schemes and sets}\label{sec:prelim}

Let $\L \subset \R^{s}$ be an $s$-dimensional lattice, i.e. $\L=\{\sum_{j=1}^sc_j\bs{l}_j:c_j\in\Z\}=\mathrm{span}_{\Z}\{\bs{l}_1,\dots,\bs{l}_s\}$ for a basis $\bs{l}_1,\dots,\bs{l}_s$ of $\R^s$.
Consider two subspaces $V_\parallel\subset\R^s$ (physical space), $V_\perp\subset\R^s$ (internal space), with $V_\parallel\oplus V_\perp=\R^s$, $\dim V_\parallel=n\geq 1$, $\dim V_\perp=s-n\geq 1$, and projections $\pi_\parallel$, $\pi_\perp$ to these subspaces. In order to simplify the formalism, we fix a basis of $\R^s$ so that writing $\bs{x}=(x_1,\dots,x_s)^\top$ we have
$$
\pi_\parallel (\bs{x}) = (x_1,\dots,x_n)^\top,\qquad \pi_\perp (\bs{x}) = (x_{n+1},\dots,x_s)^\top.
$$
In such a way, the subspaces $V_\parallel$, $V_\perp$ are just $\R^n$, $\R^{s-n}$, respectively. A cut-and-project scheme in $\R^s$ is given by the lattice $\L \subset \R^{s}$ together with the projections $\pi_\parallel,\pi_\perp$ defined above. The pair $\Lambda = (\L\subset\R^{s}, \R^n)$ is called a cut-and-project scheme.


\begin{definition}\label{de:capscheme}
	We say that a cut-and-project scheme $\Lambda=(\L\subset\R^{s}, \R^n)$ is
	\begin{itemize}
		\item[(i)] non-degenerated if $\pi_\parallel$ restricted to $\L$ is injective,
		\item[(ii)] aperiodic if $\pi_\perp$ restricted to $\L$ is injective,
		\item[(iii)] irreducible if $\pi_\perp(\L)$ is dense in $\R^{s-n}$.
	\end{itemize}
	A  cut-and-project scheme is called generic if all (i)--(iii) are satisfied.
\end{definition}

The images of the lattice $\L$ under the projections $\pi_\parallel$, $\pi_\perp$ are $\Z$-modules in $\R^n$, $\R^{s-n}$ respectively. Non-degeneracy and aperiodicity of the scheme ensure that there is a module isomorphism $\ast$ between $\pi_\parallel(\L)$ and $\pi_\perp(\L)$, namely $\ast=\pi_\perp\circ\pi_\parallel^{-1}$, which is usually called the star map.
A cut-and-project set is constructed from a generic scheme as a suitable subset of the $\Z$-module $\pi_\parallel(\L)$. The choice of the subset is directed by a bounded window
in the internal space. In order to guarantee the cut-and-project set to have reasonable properties, it is usual to take for the window a bounded set $\Omega$ whose closure $\overline{\Omega}$ is equal to the closure of its non-empty interior $\Omega^\circ$, see~\cite{Pleas}.

\begin{definition}\label{de:capset}
	Let $(\L\subset\R^s , \R^n)$ be a generic cut-and-project scheme. Given a bounded set $\Omega\subset\R^{n-s}$ such that $\overline{\Omega^\circ}=\overline{\Omega}\neq\emptyset$, ee define the cut-and-project set $\Sigma(\Omega)$ with acceptance window $\Omega$ by
	\begin{equation}
	\label{eq:cut-and-project-definice}
	\S(\Om):= \left\{ \pi_\parallel(\bs{l}) : \bs{l} \in \L \text{ and } \pi_\perp(\bs{l}) \in \Om\right\} = \{\bs{x}\in\pi_\parallel(\L): \bs{x}^\ast\in\Omega\}.
	\end{equation}
\end{definition}
The set $\Sigma(\Omega)=\Sigma_{\L}(\Omega)$ depends also on the lattice $\L$, although we usually omit the index, as the lattice is clear from the context.
With the notation of star map, one can  write
		\begin{equation}
		\label{eq:dk1}
		\left(\S_\L(\Om)\right)^\ast = \pi_\perp(\L) \cap \Om,
		\end{equation}
		and consequently
		\begin{equation}
		\label{eq:dk2}
		\overline{\left(\S_\L(\Om)\right)^\ast} = \overline{\Om},
		\end{equation}
where we use that $\pi_\perp(\L)$ is dense in $\R^{s-n}$ and the window $\Omega$ satisfies $\overline{\Om^\circ} = \overline{\Om}$.

The assumptions in Definition~\ref{de:capset} imply that $\Sigma(\Omega)$ satisfies certain properties. Non-degeneracy and irreducibility (cd.\ Definition~\ref{de:capscheme}) together with the requirements on the acceptance window $\Omega$ imply that $\Sigma(\Omega)$ is a Delone set of finite local complexity~\cite[Proposition~7.5]{Aperiodic1}.
Let us recall that a point set $\Sigma\subset\R^n$ is Delone, if there exist two radii $0<r,R<+\infty$ such that
(i) every ball of radius $R$ in $\R^n$ contains at least one element of $\Sigma$ (relative density), and
(ii) every ball of radius $r$ in $\R^n$ contains at most one element of $\Sigma$ (uniform discreteness).
A Delone set $\Sigma\subset\R^n$ is of finite local complexity if for each $\rho>0$ the set $(\Sigma-\Sigma)\cap B(\bs{z},\rho)$ is finite.

Imposing aperiodicity on the cut-and-project scheme we obtain $\Sigma(\Omega)$ which has no translational symmetry, i.e. $\bs{t}+\Sigma(\Omega)\subset \Sigma(\Omega)$  implies $\bs{t}=\bs{0}$. For, if a vector $\bs{l}\in\L$ satisfies $\pi_\perp(\bs{l})=\bs{0}$, then $\pi_\parallel(\bs{l})+\Sigma(\Omega)=\Sigma(\Omega)$. Since $\{\bs{l}\in\L : \pi_\perp(\bs{l})=\bs{0}\}$ is a sublattice of $\L$, we derive that if a set $\Sigma(\Omega)$ of~\eqref{eq:cut-and-project-definice} is constructed using a scheme which is not aperiodic (i.e.\ which does not satisfy (ii) in Definition~\ref{de:capscheme}). In our considerations, we avoid such a situation.

It can also be derived~\cite{Lagarias} that $\Sigma(\Omega)$ is a finitely generated set, which means that its $\Z$-span $\mathrm{span}_\Z\{\Sigma(\Omega)\}$ is a finitely generated $\Z$-module. Moreover, $\mathrm{span}_\Z\{\Sigma(\Omega)\}=\pi_\parallel(\L)$, thus the $\Z$-module is of rank $s$.

Throughout the paper, it turns suitable to use the following matrix formalism for the above definitions.
For the $j$-th column of a matrix $T$ we write $T_{\bullet j}$, the $j$-th row is denoted by $T_{j\bullet}$.
We write the columns  of vector generators $\bs{l}_j$ of the lattice $\L$ into a (non-singular) matrix $L\in\R^{s\times s}$, i.e. $L_{\bullet j}=\bs{l}_j$ for $j=1,\dots,s$. The lattice $\L$ thus can be written as $\L=\{L\bs{r}: \bs{r}\in\Z^s\}=\mathrm{span}_\Z\{L_{\bullet 1},\dots,L_{\bullet s}\}$. We say that $L$ is a matrix associated to the lattice $\L$. Note that the associated matrix $L$ is not unique. For, the choice of the lattice base is given only up to a transformation by an integer matrix of determinant $\pm1$.

The action of projections $\pi_\parallel$, $\pi_\perp$ can also be written in a matrix form. Given two matrices $T_1,T_2$ with the same number of rows, we write $(T_1,T_2)$ for the matrix arising by putting the columns of $T_2$ after the columns of $T_1$. Similarly, we write $\binom{T_1}{T_2}$ for matrices $T_1,T_2$ with the same number of columns.

With this notation, the action of $\pi_\parallel$, $\pi_\perp$ is written as application of the $(n\times s)$-matrix $(I_n,O)$, and $\left((s-n)\times s\right)$-matrix $(O, I_{s-n})$, respectively, i.e. we have
$$
\begin{aligned}
\pi_\parallel(\L)&=(I_n,O)\L=\{(I_n,O)L\bs{r}:\bs{r}\in\Z^s\},\\
\pi_\perp(\L)&=(O,I_{s-n})\L=\{(O,I_{s-n})L\bs{r}:\bs{r}\in\Z^s\},
\end{aligned}
$$
where $I_j$ stands for the identity matrix of order $j$ and $O$ is the zero matrix of suitable dimension.

It follows from results and criterions (the so-called V-condition, W-condition) derived by Pleasants \cite{Pleas} that one can rewrite terms as irreducibility, aperiodicity and non-degeneracy using the associated matrix $L$ as follows:

\begin{prop}\label{p:pleasants}
	Let $\Lambda = (\L\subset \R^s, \R^n)$ be a cut-and-project scheme. Let $L \in \R^{s\times s}$ be the matrix associated to $\L$. Then $\Lambda$ is
	\begin{itemize}
		\item[(i)] non-degenerated if and only if for all $\bs{x}\in \R^{s-n}$ it holds that \[\begin{pmatrix}
		\bs{0} \\ \bs{x}
		\end{pmatrix} \in \L \Rightarrow \bs{x} = \bs{0}, \]
		\item[(ii)] aperiodic if and only if for all $\bs{x}\in \R^{n}$ it holds that \[\begin{pmatrix}
		\bs{x} \\ \bs{0}
		\end{pmatrix} \in \L \Rightarrow \bs{x} = \bs{0}. \]
        \item[(iii)] irreducible if and only if for all  $S \in \Z^{s\times (s-1)}$ there exists $\bs{x}\in \R^n$ such that
		\[\begin{pmatrix}
		\bs{x} \\ \bs{0} \end{pmatrix} \notin \left\{ LS\bs{u}: \bs{u}\in \R^{s-1} \right\},
		\]
	\end{itemize}
\end{prop}

\begin{proof}
Items (i) and (ii) are obvious from the definition of non-degeneracy and aperiodicity. In order to prove irreducibility,
we use the statement of Pleasants~\cite[Corollary 2.12]{Pleas}, who shows that irreducibility of the scheme is equivalent to the fact that the physical space is not contained in
any hyperplane generated by lattice vectors. This fact is expressed in Item (iii).
\end{proof}

\section{Self-similarities of a cut-and-project scheme}\label{sec:4}

The purpose of this article is to study generalized self-similarities of cut-and-project sets. By a~self-similarity of a set $\Sigma\subset\R^n$ we understand a linear mapping $A:\R^n\to\R^n$ such that $A\Sigma\subset\Sigma$. Let us recall the main Question~\ref{q1} which we plan to solve.
Given a linear mapping $A:\R^n\to\R^n$, our aim is to decide whether there exists a cut-and-project set $\Sigma(\Omega)\subset\R^n$ with $A$ as its self-similarity.
If yes, then to determine the minimal dimension $s$ of the corresponding cut-and-project scheme and describe the construction.

As the first step, realize that self-similarity of the cut-and-project set $\Sigma(\Omega)\subset\R^n$ implies self-similarity of the corresponding $\Z$-module $\pi_\parallel(\L)$. We have the following statement.

\begin{prop}\label{p:q1->q2}
	Let $\Sigma(\Omega)\subset\R^n$ be a cut-and-project set constructed from a generic cut-and-project scheme $(\L\subset\R^s,\R^n)$. Let $L\in\R^{s\times s}$ be a matrix associated to the lattice $\L$. If $A:\R^n\to\R^n$ is a linear mapping such that $A\Sigma(\Omega)\subset\Sigma(\Omega)$, then $A\pi_\parallel(\L)\subset\pi_\parallel(\L)$ and, moreover, there exist uniquely defined matrices $C\in \Z^{s \times s}$ and $B \in \R^{ (s-n) \times (s-n)}$ such that
	\begin{equation}\label{eq:ABC}
	\begin{pmatrix}
	A & O \\
	O & B
	\end{pmatrix} L = LC.
	\end{equation}
\end{prop}

\begin{proof}
It can be shown that $\mathrm{span}_\Z\{\Sigma(\Omega)\}=\pi_\parallel(\L)$, see~\cite{Moody}. Consequently, one can easily derive that $A\Sigma(\Omega)\subset\Sigma(\Omega)$ implies
$A\pi_\parallel(\L)\subset\pi_\parallel(\L)$.

Let $\bs{l}_1,\dots,\bs{l}_s$ be vectors generating the lattice $\L$ and denote $\bs{x}_i=\pi_\parallel(\bs{l}_i)$. Since $A\bs{x}_i\in\pi_\parallel(\L)$, there exists a lattice vector $\tilde{\bs{l}}_i$ such that $A\bs{x}_i=\pi_\parallel(\tilde{\bs{l}}_i)$. As any lattice vector is an integer combination of $\bs{l}_1,\dots,\bs{l}_s$, there exists $\bs{c}_i\in\Z^s$ such that
$\tilde{\bs{l}}_i=L\bs{c}_i$. The assignment $\bs{l}_i\mapsto \tilde{\bs{l}}_i$ is a linear map whose matrix in the basis $\bs{l}_1,\dots,\bs{l}_s$ is $C\in\Z^{s\times s}$ composed from the columns $\bs{c}_1,\dots,\bs{c}_s$.
Therefore, each linear map $A$ of a non-degenerate cut-and-project scheme $(\L\subset \R^s, \R^{n})$ into itself is associated with some integer matrix $C$ through the relation $A(I_n, O)L = (I_n,O)LC$. The matrix $C$ is unique due to the injectivity of $\pi_\parallel$.

Denote by $F,B$ the matrices $F\in\R^{n\times(s-n)}$, $B\in\R^{(s-n)\times(s-n)}$ such that $(O,I_{s-n})LC=(F,B)L$. It suffices to set $(F,B)=(O,I_{s-n})LCL^{-1}$.
We then have
\begin{equation}\label{eq:ABCF}
\begin{pmatrix}
	A & O \\
	F & B
\end{pmatrix} L = LC
\end{equation}

In order to complete the proof, it remains to show that the matrix $F$ is a zero matrix. For the contradiction, suppose that $F$ has a non-zero row, say $\bs{f}^\top=(F_{i1},\dots,F_{in})\neq \bs{0}$. Then there exists $\bs{x}_0\in\Sigma(\Omega)$ such that $\bs{f}^\top\bs{x}_0\neq 0$. Indeed, otherwise $\Sigma(\Omega)$ belongs to a hyperplane in $\R^n$ which contradicts the fact that $\Sigma(\Omega)$ is relatively dense in $\R^n$. For the element $\bs{x}_0\in\Sigma(\Omega)$ find $\bs{r}_0\in\Z^{s}$ such that $\bs{x}_0=\pi_\parallel(L\bs{r}_0)$. We have $\bs{f}^\top(I_n,O)L\bs{r}_0\neq 0$.

Relative density of $\Sigma(\Omega)$ further implies existence of a radius $R$ such that every ball $B(\bs{z},R)$, $\bs{z}\in\R^n$ has a non-empty intersection with $\Sigma(\Omega)$.
Therefore, for every positive integer $m$ the ball $B(m\bs{x}_0,R)$ contains at least one point $\bs{x}_m$ of $\Sigma(\Omega)$. In other words, for every $m\in\N$ there exists $\bs{z}_m\in\R^n$, $\|\bs{z}_m\|<R$, such that $\bs{x}_m=m\bs{x}_0+\bs{z}_m\in\Sigma(\Omega)$. We thus have a sequence of integer vectors $\bs{r}_m\in\Z^s$ such that
$$
\begin{aligned}
\bs{x}_m&=\pi_\parallel(L\bs{r}_m) = (I_{n},O)L\bs{r}_m\in\Sigma(\Omega)\\
 \text{ and thus }\ \bs{x}_m^\ast &=\pi_\perp(L\bs{r}_m) =(O, I_{s-n})L\bs{r}_m\in \Omega.
\end{aligned}
$$
Then $A(\bs{x}_m)=A\pi_\parallel(L\bs{r}_m) = (I_{n},O)LC\bs{r}_m\in\Sigma(\Omega)$ and therefore by~\eqref{eq:ABCF}
$$
A(\bs{x}_m^\ast)=(O, I_{s-n})LC\bs{r}_m = (F,B)L\bs{r}_m\in \Omega \quad \text{ for every $m\in\N$.}
$$
Since $\Omega$ is a bounded set in $\R^{s-n}$, the $i$-th row of the vector $(F,B)L\bs{r}_m=F(I_n,O)L\bs{r}_m+B(0,I_{s-n})L\bs{r}_m)$ is also bounded, say, by a constant $\gamma$. We write
$$
\big|\bs{f}^\top (I_n,O)L\bs{r}_m + \bs{b}^\top (O,I_{s-n})L\bs{r}_m\big|<\gamma\,,
$$
where we have denoted $\bs{b}^\top=(B_{i1},\dots,B_{i,s-n})$ the $i$-th row of the matrix $B\in\R^{(s-n)\times(s-n)}$.
Using the above, we have
$$
\big|m\bs{f}^\top\bs{x}_0+\bs{f}^\top\bs{z}_m + \bs{b}^\top (O,I_{s-n})L\bs{r}_m\big|<\gamma
$$
and consequently
$$
m|\bs{f}^\top\bs{x}_0|< \gamma + |\bs{f}^\top\bs{z}_m| + |\bs{b}^\top (O,I_{s-n})L\bs{r}_m|.
$$
The right hand side of the latter is bounded, as $|\bs{f}^\top\bs{z}_m|\leq \|\bs{f}\|\|\bs{z_m}\|<R\|\bs{f}\|$ and $(O,I_{s-n})L\bs{r}_m\in\Omega$.
On the other hand, since $\bs{f}^\top\bs{x}_0\neq 0$, the left hand side tends to infinity with $m\to\infty$. This is a contradiction. Necessarily,
the $i$-th row of the matrix $F$ vanishes for every $i=1,\dots,s-n$ and thus $F$ is a zero matrix, as we intended to show.
\end{proof}

The above proposition suggests that the first step in recognizing cut-and-project sets with self-similarity $A$ is to find self-similar cut-and-project schemes.
The notion of a self-similarity of a cut-and-project scheme includes not only the fact that the $\Z$-module $\pi_\parallel(\L)$ is closed under the action of $A$, but also that
$A$ induces a linear action of the internal space, i.e.\ a linear map $B$ on $\R^{s-n}$ such that $B\pi_\perp(\L)\subset\pi_\perp(\L)$. Formally, it is described by the following definition.

\begin{definition}\label{de:self}
Let $\Lambda=(\L\subset\R^s,\R^n)$ be a generic cut-and-project scheme and let $L\in\R^{s\times s}$ be a matrix associated to the lattice $\L$.
We say that a linear mapping $A:\R^n\to\R^n$ is a self-similarity of the scheme  $\Lambda$ if
\begin{itemize}
  \item[(1)] $A\pi_\parallel(\L)\subset\pi_\parallel(\L)$,
  \item[(2)] there exist matrices $C\in \Z^{s \times s}$ and $B \in \R^{ (s-n) \times (s-n)}$ such that~\eqref{eq:ABC} is valid.
\end{itemize}
\end{definition}

With this we formulate:

\begin{que}\label{q2}
	To a given linear mapping $A:\R^n\to\R^n$, decide whether there exists a generic cut-and-project scheme $\Lambda=(\L\subset\R^s,\R^n)$ with self-similarity $A$. If yes, determine the minimal dimension $s$ of such a cut-and-project scheme and describe the construction.
\end{que}

From~\eqref{eq:ABC}, it is obvious that the matrix $C$ is similar to a block diagonal matrix with blocks $A$, $B$ on the diagonal. Thus for the spectra of these matrices we have $\sigma(C)=\sigma(A)\cup\sigma(B)$. Since $C$ is an integer matrix, its characteristic polynomial is monic with integer coefficients. We have thus derived the following corollary.

\begin{coro}
	\label{c:duseldek1}
	Let $A\in\R^{n\times n}$ be a self-similarity of a generic cut-and-project scheme $(\Lambda$. Then
	the eigenvalues of the matrix $A$ are algebraic integers and their minimal polynomials over $\Q$ divide the characteristic polynomial of the matrix $C$ from~\eqref{eq:ABC}.
\end{coro}

For the construction of self-similar cut-and-project schemes we will use the following proposition.

\begin{prop}\label{p:ABCzpet}
	Let $A\in\R^{n\times n}$, $B\in\R^{(s-n)\times (s-n)}$, $C\in\Z^{s\times s}$ and let $L\in\R^{s\times s}$ be non-singular such that~\eqref{eq:ABC}. Then the $\Z$-module $\pi_\parallel(\L)$, where $\L=\mathrm{span}_\Z\{L_{\bullet 1},\dots,L_{\bullet s}\}$, and $\pi_\parallel(\L)=(I_n,O)\L$ satisfies $A\pi_\parallel(\L)\subset \pi_\parallel(\L)$. If, moreover, the scheme $(\L\subset\R^s,\R^n)$ is generic, then $A$ is its self-similarity.
\end{prop}

In order to prove non-degeneracy, aperiodicity and irreducibility  of the cut-and-project scheme $(\L\subset\R^s,\R^n)$  implicitly defined in Proposition~\ref{p:ABCzpet}, we will use tools given in Proposition~\ref{p:generic}.

\section{Minimal polynomials over $\Q$}
\label{sec:minimal_polynomial_estimation}

This section shows some necessary conditions on the minimal polynomials over $\Q$ of matrices $A$, $B$ and $C$ defining a self-similarity of a cut-and-project scheme by \eqref{eq:ABC}. This provides a useful information on the relation of spectra of these matrices.

\begin{lem}
	\label{l:nulpol}
Let $\Lambda=(\L\subset\R^s,\R^n)$ be a generic cut-and-project scheme with self-similarity $A\in\R^{n\times n}$ and let $B\in\R^{(s-n)\times(s-n)}$ be the matrix from~\eqref{eq:ABC}.
Then for every polynomial $p\in \Z[X]$ it holds that  $p(A) = O$ if and only if $p(B) = O$.
\end{lem}

\begin{proof}
	According to \eqref{eq:ABC}, for any polynomial $p$ it holds that
	\[L p(C) L^{-1} = p(LCL^{-1}) = \begin{pmatrix}
	p(A) & O \\ O & p(B)
	\end{pmatrix}. \]
	
	For a contradiction suppose that $p(A) = O$ and $p(B) \neq O$ (so consequently $p(C) \neq O$). Then
	\[L p(C) = \begin{pmatrix}
	O & O \\ O & p(B)
	\end{pmatrix}L. \]
	Since $p(C) \in \Z^{s\times s}$ is a non-zero matrix there exists a non-zero column in $p(C)$, denote it $\bf{r}$. Then the vector $\bs{\ell}:= L\bf{r}\in \L$ is a non-zero lattice vector whose first projection vanishes, i.e.\ $\pi_\parallel(\bs{\ell}) = \bs{0}$. Thus the restriction of $\pi_\parallel$ to $\L$ is not injective and this is a contradiction with non-degeneracy of the scheme $\Lambda$. The opposite implication is proved in the same way using aperiodicity of $\Lambda$.
\end{proof}

Combining the above lemma together with properties of minimal polynomials, we obtain the
necessary conditions for creating generic cut-and-project schemes.

\begin{prop}
	\label{p:minpol}
	Let $\Lambda=(\L\subset \R^s,\R^n)$ be a generic cut-and-project scheme with self-similarity $A$. Let $B,C$ be the matrices from Definition~\ref{de:self}. Then $\m{A} = \m{B} = \m{C}= \mu_C$. In particular, either $A,B,C$ are all non-singular, or all singular matrices. Moreover, $A,B,C$ are all diagonalizable over $\C$ or all non-diagonalizable.
\end{prop}

\begin{proof}
According to Lemma~\ref{l:nulpol}, it holds that $\m{A}=\m{B}=\m{C}$.
Since $C$ is a rational matrix, we moreover have $\mu_C=\m{C}$. The fact that all the three matrices $A,B,C$ have the same minimal polynomial over $\Q$ implies that non-singularity and diagonalizability over $\C$ is valid for all of them or none of them.
\end{proof}

\begin{example}
Requiring a cut-and-project scheme with $k$-fold rotation  symmetry $A$, then necessary $A^k = I$. This implies that $A$ is annihilated by the polynomial $X^k - 1$. The polynomial $X^k -1$ is divisible by the cyclotomic polynomial $\Phi_k(X) \in \Z[X]$. $\Phi_k(X)$ is irreducible and minimal one over $\Q$ that annihilates $A$. Using the condition derived in Proposition~\ref{p:minpol} one gets
\[\m{A} = \m{C} = \mu_C. \]
Taking $C$ such that the minimal polynomial $\mu_C$ is the characteristic polynomial $\chi_C$ (for example the companion matrix to $\Phi_k$) we have an estimation on the minimal dimension of $C$ given by $\phi(k)=\deg(\Phi_k)$, which is in agreement with the statement in~\cite{Aperiodic1}.
\end{example}

\begin{coro}\label{c:spektrumAB}
Let $\Lambda$ be a generic cut-and-project scheme with self-similarity $A$ and let $B$ be the matrix from Definition~\ref{de:self}. Then each eigenvalue of $A$
 is an algebraic conjugate of an eigenvalue of $B$ and vice versa.
\end{coro}

\section{Transformations of cut-and-project schemes}

The following statement allows us to reduce our attention in~\eqref{eq:ABC} and~\ref{p:ABCzpet} to matrices in a special form.

\begin{lem}\label{l:BUNO}
	Let $\Lambda=(\L\subset\R^s,\R^n)$ be a generic cut-and-project scheme and let $\L\subset\R^s$ be a lattice with an associated matrix $L\in\R^{s\times s}$.
	Let further $W_A\in\R^{n\times n}$, $W_B\in\R^{(s-n)\times (s-n)}$, $Q\in\Q^{s\times s}$ be non-singular matrices. Define
  \[
    \widetilde{L}=\begin{pmatrix}
    W_A & O \\
	O & W_B
	\end{pmatrix}LQ^{-1} \mbox{ and } \widetilde{\L}=\mathrm{span}_\Z\{\widetilde{L}_{\bullet 1},\dots,\widetilde{L}_{\bullet s}\}.
  \]
Then $\widetilde{\Lambda}=(\widetilde{\L}\subset\R^s,\R^n)$ is also a generic cut-and-project scheme.
\end{lem}

\begin{proof}
	In proving non-degeneracy of the scheme $\widetilde{\Lambda}$ we will proceed with the help of Proposition~\ref{p:pleasants}.
Consider a vector of the lattice $\widetilde{\L}$, i.e. of the form $\widetilde{L}\widetilde{\bs{r}}$ for some integer vector $\widetilde{\bs{r}}\in\Z^s$.
Denoting $\bs{r}:=Q^{-1}\widetilde{\bs{r}}$, we have
$$
\widetilde{L}\widetilde{r} = \begin{pmatrix}
    W_A & O \\
	O & W_B
	\end{pmatrix}LQ^{-1} \widetilde{\bs{r}} = \begin{pmatrix}
    W_A & O \\
	O & W_B
	\end{pmatrix} L\bs{r}.
$$
If $\widetilde{L}\widetilde{\bs{r}}=\binom{\bs{0}}{\widetilde{\bs{x}}}\in\widetilde{\L}$, then
$$
L\bs{r} = \begin{pmatrix}
    W_A^{-1} & O \\
	O & W_B^{-1}
	\end{pmatrix} \binom{\bs{0}}{\widetilde{\bs{x}}} = \binom{\bs{0}}{W_B^{-1}\widetilde{\bs{x}}}\in\L.
$$
Since $W_B$ is a non-singular matrix, $W_B^{-1}\widetilde{\bs{x}}=\bs{0}$ is equivalent to $\widetilde{\bs{x}}=\bs{0}$. This completes the proof of non-degeneracy. Similarly we proceed to show aperiodicity of the scheme $(\widetilde{\L}\subset\R^s,\R^n)$.
	
	In order to prove irreducibility of the cut-and-project scheme $(\widetilde{\L}\subset\R^s,\R^n)$, we need to show that $\pi_\perp(\widetilde{\L})$ is dense in $\R^{s-n}$.
	Note that $\pi_\perp(\L)=W_B^{-1}\pi_\perp(\widetilde{\L}Q)$. Denoting $q\in\N$ such that $qQ\in\Z^{s\times s}$, we have $\widetilde{\L}Q=\frac1q \widetilde{\L}(qQ)\subset \frac1q \widetilde{\L}$, and consequently
	$q W_B \pi_\perp(\L)=\pi_\perp(\widetilde{\L}Q) \subset \pi_\perp(\widetilde{\L})$. Therefore
	density of $\pi_\perp(\L)$ in $\R^{s-n}$ implies density of $\pi_\perp(\widetilde{\L})$ in $\R^{s-n}$.
	%
\end{proof}

\begin{definition}
	We will say that cut-and-project schemes $\Lambda$ and $\widetilde{\Lambda}$ from Lemma~\ref{l:BUNO} are equivalent.
\end{definition}

Now we need to introduce the following notation.
Let $\mathbb{F}$ be a field. If $M, \widetilde{M}$ are square matrices of the same order, say $k$, we say that $M$ is similar to $\widetilde{M}$ over $\mathbb{F}$, denoted by $M \sim_{\mathbb{F}} \widetilde{M}$, if there exists a non-singular matrix $F\in \mathbb{F}^{k\times k}$ such that $MF = F\widetilde{M}$.

\begin{lem}
	\label{l:fn}
	Let $A \in \R^{n \times n}$ be a self-similarity of a generic cut-and-project scheme $\Lambda = \left(\L\subset \R^s, \R^n \right)$ with an associated matrix $L \in \R^{s\times s}$, i.e. $A \pi_\parallel (\L) \subset \pi_\parallel (\L)$. Let $B\in \R^{(s-n)\times (s-n)}$ and $C \in \Z^{s\times s}$ be as in Definition~\ref{de:self}, i.e. equation \eqref{eq:ABC} holds. Take $\widetilde{A},\ \widetilde{B}, \ \widetilde{C}$ arbitrary matrices such that $A\sim_\R \widetilde{A}$, $ B \sim_\R \widetilde{B}$, $ C \sim_\Q \widetilde{C}$.
	Then there exists a generic cut-and-project scheme $\widetilde{\Lambda} = \left(\widetilde{\L}\subset \R^{s},\R^n \right)$  such that $\widetilde{A} \pi_\parallel(\widetilde{\L}) \subset \pi_\parallel(\widetilde{\L})$ and
	\begin{equation}
	\label{eq:ABC-tilda}
	\begin{pmatrix}
	\widetilde{A} & O \\ O & \widetilde{B}
	\end{pmatrix} \widetilde{L} = \widetilde{L} \widetilde{C},
	\end{equation}
where $\widetilde{L}\in \R^{s\times s}$ is a matrix associated to $\widetilde{\L}$.
\end{lem}

\begin{proof}
	Let $W_A\in \R^{n\times n}$, $W_B \in \R^{(s-n)\times (s-n)}$ and $Q \in \Q^{s \times s}$ be matrices such that  $\widetilde{A}=W_AAW_A^{-1}$, $\widetilde{B}=W_BBW_B^{-1}$, $\widetilde{C}=QCQ^{-1}$. Set $\widetilde{L}=\left(\begin{smallmatrix}                                                                                                                                               W_A & O \\
	O & W_B
	\end{smallmatrix}\right)LQ^{-1}$. Multiplying \eqref{eq:ABC} from the left side by
$\left(\begin{smallmatrix}                                                                                                                                               W_A & O \\
	O & W_B
	\end{smallmatrix}\right)$ and from the right side by $Q^{-1}$, we get
\begin{multline*}
\begin{pmatrix}
	W_A & O \\ O & W_B
\end{pmatrix}\begin{pmatrix}
	A & O \\ O& B
\end{pmatrix} \begin{pmatrix}
	W_A^{-1} & O \\ O & W_B^{-1}
	\end{pmatrix}
\begin{pmatrix}
	W_A & O \\ O & W_B
	\end{pmatrix} LQ^{-1} =\\
= \begin{pmatrix}
	W_A & O \\ O & W_B
\end{pmatrix} L Q^{-1} QCQ^{-1},
\end{multline*}
which is relation \eqref{eq:ABC-tilda}. Lemma \ref{l:BUNO} gives genericity of $\widetilde{\Lambda}$.
By Proposition~\ref{p:ABCzpet}, we have that $\widetilde{A} \pi_\parallel(\widetilde{\L}) \subset \pi_\parallel(\widetilde{\L})$.
\end{proof}

As a consequence of the above considerations we realize that when solving Question~\ref{q2}, we can, without loss of generality, consider~\eqref{eq:ABC} with matrices $A\in\R^{n\times n}$, $B\in\R^{(s-n)\times(s-n)}$ taken in the real Jordan form and matrix $C\in\Z^{s\times s}$ in the rational Jordan form.

\begin{coro}\label{c:BUNO}
Let $A\in\R^{n \times n}$ and let $\widetilde{A}$ be its real Jordan form. Then the answer to Question~\ref{q2} for $A$ is yes if and only if the answer to Question~\ref{q2} with $\widetilde{A}$ is yes. If it is the case, then the minimal dimensions of generic cut-and-project schemes for $A$ and for $\widetilde{A}$ coincide.
\end{coro}

\section{Composition of cut-and-project schemes}

When constructing a cut-and-project scheme with a given self-similarity, we will proceed by analyzing the spectrum of $A$. We will compose the schemes found for the minimal polynomials of the individual eigenvalues. Such elementary cases will be described in Section~\ref{sec:element}. First, let us explain more precisely what we mean by composition of schemes.

\begin{definition}\label{de:directsum}
  Let $\hat{\Lambda}=(\hat{\L}\subset\R^{\hat{s}},\R^{\hat{n}})$, $\check{\Lambda}=(\check{\L}\subset\R^{\check{s}},\R^{\check{n}})$ be cut-and-project schemes and let $\hat{L}$, $\check{L}$ be matrices in $\R^{\hat{s}\times\hat{s}}$, $\R^{\check{s}\times\check{s}}$ associated to the lattices $\hat{\L}$, $\check{\L}$, respectively.
  Set
  $$
  L=P\begin{pmatrix}
  \hat{L} & O \\
  O & \check{L}
  \end{pmatrix}
  \quad\text{ where }\quad
  P=\begin{pmatrix}
      I_{\hat{n}} & O & O & O \\
      O & O & I_{\hat{s}-\hat{n}} & O \\
      O & I_{\check{n}} & O & O \\
      O & O & O & I_{\check{s}-\check{n}}
    \end{pmatrix}\,.
  $$
  Denote $s=\hat{s}+\check{s}$, $n=\hat{n}+\check{n}$. Define the lattice $\L\in\R^s$ by $\L:=\mathrm{span}_\Z\{L_{\bullet 1},\dots, L_{\bullet s}\}$
  and projections $\pi_\parallel:\R^s\to\R^n$, $\pi_\perp:\R^s\to\R^{s-n}$ by $\pi_\parallel(\L)=(I_n,O)\L$, $\pi_\perp(\L)=(O,I_{s-n})\L$. Then the scheme $\Lambda=(\L\subset\R^s,\R^n)$ is called the direct sum
  of schemes $\hat{\Lambda}$, $\check{\Lambda}$. 
  We denote it by $\Lambda=\hat{\Lambda}\oplus\check{\Lambda}$. 
\end{definition}

Inductively, we define direct sum of more than two cut-and-project schemes.

\begin{lem}\label{l:directsumgeneric}
  Let $\Lambda=\hat{\Lambda}\oplus\check{\Lambda}$. Then
  \begin{itemize}
    \item $\Lambda$ is non-degenerate if and only if $\hat{\Lambda}$ and $\check{\Lambda}$ are both non-degenerate;
    \item   $\Lambda$ is aperiodic if and only if $\hat{\Lambda}$ and $\check{\Lambda}$ are both aperiodic;
    \item $\Lambda$ is irreducible if and only if $\hat{\Lambda}$ and $\check{\Lambda}$ are both irreducible.
  \end{itemize}
\end{lem}

\begin{proof}
  Let $\Lambda=(\L\subset\R^s,\R^n)$, $\hat{\Lambda}=(\hat{\L}\subset\R^{\hat{s}},\R^{\hat{n}})$, $\check{\Lambda}=(\check{\L}\subset\R^{\check{s}},\R^{\check{n}})$, and let
  $L, \hat{L}, \check{L}$ be the matrices associated to lattices $\L, \hat{\L}, \check{\L}$, respectively.
  Write $L=\binom{L_1}{L_2}$ where $L_1\in\R^{n\times s}$, $L_2\in\R^{(s-n)\times s}$.
  Similarly, we introduce notation for $\hat{L}=\binom{\hat{L}_1}{\hat{L}_2}$ and $\check{L}=\binom{\check{L}_1}{\check{L}_2}$.
  In the proof, we use Proposition~\ref{p:pleasants}. Any vector $\bs{l}\in\L$ is written as $\bs{l}=L\bs{r}=L\binom{\hat{\bs{r}}}{\check{\bs{r}}}$ for some $\hat{\bs{r}}\in\Z^{\hat{s}}$, $\check{\bs{r}}\in\Z^{\check{s}}$, and we have the corresponding lattice vectors $\hat{\bs{l}}=\hat{L}\hat{\bs{r}}\in\hat{L}$, $\check{\bs{l}}=\check{L}\check{\bs{r}}\in\check{L}$.
  We have
  $$
   \bs{l}=
   \begin{pmatrix}
      \hat{L}_{1} & O \\
      O & \check{L}_{1} \\
      \hat{L}_{2} & O \\
      O & \check{L}_{2}
    \end{pmatrix}\binom{\hat{\bs{r}}}{\check{\bs{r}}} =
    \begin{pmatrix}
      \hat{L}_{1}\hat{\bs{r}} \\
      \check{L}_{1}\check{\bs{r}} \\
      \hat{L}_{2} \hat{\bs{r}}\\
      \check{L}_{2} \check{\bs{r}}
    \end{pmatrix} = \binom{\bs{0}}{\bs{x}} \quad\iff\quad
      \hat{\bs{l}}=\binom{\bs{0}}{\hat{\bs{x}}},\ \check{\bs{l}}=\binom{\bs{0}}{\check{\bs{x}}}
  $$
where $\bs{x}=\binom{\hat{\bs{x}}}{\check{\bs{x}}}$. Therefore by Item (i) of Proposition~\ref{p:pleasants}, the scheme $\Lambda$ is non-degenerate if and only if $\hat{\Lambda}$ and $\check{\Lambda}$ are both non-degenerate.
Similarly we proceed for proving aperiodicity.

  For irreducibility, we use the fact that cartesian product of sets $\pi_\perp(\hat{\L})$, $\pi_\perp(\check{\L})$ that are dense in $\R^{\hat{s}-\hat{n}}$, $\R^{\check{s}-\check{n}}$ resp., is dense in $\R^{s-n}$ and vice versa.
\end{proof}

\begin{lem}
	\label{l:sklad}
  Let $\hat{A}\in\R^{\hat{n}\times \hat{n}}$, $\check{A}\in\R^{\check{n}\times\check{n}}$ be self-similarities of generic cut-and-project schemes  $\hat{\Lambda}$, $\check{\Lambda}$, respectively. Then $A=\left(\begin{smallmatrix}
      \hat{A} & O \\
      O & \check{A}
    \end{smallmatrix}\right)$,
is a self-similarity of the generic cut-and-project scheme $\Lambda=\hat{\Lambda}\oplus\check{\Lambda}$.
\end{lem}

\begin{proof}
  By Definition~\ref{de:self}, there exist matrices $\hat{B}\in\R^{(\hat{s}-\hat{n})\times (\hat{s}-\hat{n})}$, $\hat{C}\in\Z^{\hat{s}\times\hat{s}}$, $\check{B}\in\R^{(\check{s}-\check{n})\times(\check{s}-\check{n})}$, $\check{C}\in\Z^{\check{s}\times\check{s}}$ such that
  $$
  \begin{pmatrix}
      \hat{A} & O \\
      O & \hat{B}
   \end{pmatrix}\hat{L}=\hat{L}\hat{C}\,,\qquad
  \begin{pmatrix}
      \check{A} & O \\
      O & \check{B}
   \end{pmatrix}\check{L}=\check{L}\check{C}\,,
  $$
where $\hat{L}\in\R^{\hat{s}\times\hat{s}}$, $\check{L}\in\R^{\check{s}\times\check{s}}$ are matrices associated to the lattices $\hat{\L}\subset\R^{\hat{s}}$, $\check{\L}\subset\R^{\check{s}}$. Denoting as above
$\hat{L}=\binom{\hat{L}_1}{\hat{L}_2}$, $\check{L}=\binom{\check{L}_1}{\check{L}_2}$, we derive
$$
\begin{pmatrix}
  \hat{A} & O & O & O \\
  O & \check{A} & O & O \\
  O & O & \hat{B} & O \\
  O & O & O & \check{B}
\end{pmatrix}
\begin{pmatrix}
      \hat{L}_{1} & O \\
      O & \check{L}_{1} \\
      \hat{L}_{2} & O \\
      O & \check{L}_{2}
    \end{pmatrix}=\begin{pmatrix}
      \hat{L}_{1} & O \\
      O & \check{L}_{1} \\
      \hat{L}_{2} & O \\
      O & \check{L}_{2}
    \end{pmatrix}\begin{pmatrix}
      \hat{C} & O \\
      O & \check{C}
    \end{pmatrix}\,.
$$
Setting $B=\left(\begin{smallmatrix}
      \hat{B} & O \\
      O & \check{B}
    \end{smallmatrix}\right)$,
    $C=\left(\begin{smallmatrix}
      \hat{C} & O \\
      O & \check{C}
    \end{smallmatrix}\right)$,
     we have $\Big(\begin{smallmatrix}
      {A} & O \\
      O & {B}
   \end{smallmatrix}\Big){L}={L}{C}$, for the matrix
   $L=P\left(\begin{smallmatrix}
  \hat{L} & O \\
  O & \check{L}
  \end{smallmatrix}\right)$ corresponding to the direct sum of schemes in Definition~\ref{de:directsum} which is generic by Lemma~\ref{l:directsumgeneric}. The statement follows from Proposision~\ref{p:ABCzpet}.
\end{proof}

We would like to show that if $\Lambda$ is a cut-and-project scheme with self-similarity $A$ and $A$ can be decomposed into a block-diagonal form $A=\left(\begin{smallmatrix} \hat{A} & O \\ O & \check{A}\end{smallmatrix}\right)$, then $\Lambda$ is a direct sum of smaller schemes with self-similarities $\hat{A}$, $\check{A}$, respectively.
This is however not true in general. To see it, it suffices to realize that any cut-and-project scheme has $A= I_{\hat{n}+\check{n}}$ as a self-similarity, even though such a scheme may not be decomposable as a direct sum. Thus, we have to impose a restriction on the spectra of the self-similarities.

\begin{prop}\label{p:Anekonjug}
	Let $A\in\R^{n\times n}$ be a matrix such that all its eigenvalues are algebraic integers.
	Let $\hat{A}\in\R^{\hat{n}\times\hat{n}}$, $\check{A}\in\R^{\check{n}\times\check{n}}$ be such that $A=\left(\begin{smallmatrix} \hat{A} & O \\ O & \check{A} \end{smallmatrix}\right)$.
	Moreover, suppose that no eigenvalue of $\hat{A}$ is an algebraic conjugate of any eigenvalue of $\check{A}$.
	Let $A$ be a self-similarity of a~generic cut-and-project scheme $\Lambda=(\L\subset\R^s,\R^n)$. Then there exist generic cut-and-project schemes
	$\hat{\Lambda}=(\hat{\L}\subset\R^{\hat{s}},\R^{\hat{n}})$, $\check{\Lambda}=(\check{\L}\subset\R^{\check{s}},\R^{\check{n}})$ such that
	$A$ is a self-similarity of the direct sum $\hat{\Lambda}\oplus\check{\Lambda}$.
\end{prop}


\begin{proof}
	Let $\Lambda=(\L\subset\R^s,\R^n)$ and let $L\in\R^{s\times s}$ be a non-singular matrix associated to $\L$.
	Let $A=\left(\begin{smallmatrix}
	\hat{A} & O \\
	O & \check{A}
	\end{smallmatrix}\right)$ for $\hat{A}\in\R^{\hat{n}\times\hat{n}}$, $\check{A}\in\R^{\check{n}\times\check{n}}$. If $A$ is a self-similarity of $\pi_\parallel(\L)$, then by Definition~\ref{de:self}, there exist $C\in\Z^{s\times s}$, $B\in\R^{(s-n)\times (s-n)}$ such that $\left(\begin{smallmatrix}
	A & O \\
	O & B
	\end{smallmatrix}\right)L=LC$. According to Corollary~\ref{c:spektrumAB}, the spectrum of $B$ can be decomposed into a union of a set of numbers which are algebraic conjugates of eigenvalues in the spectrum of $\hat{A}$ and a set of numbers conjugated to eigenvalues of $\check{B}$. These sets are non-empty and disjoint. Therefore without loss of generality (by Lemma~\ref{l:BUNO})
the matrices $B$ and $C$ are of the form $B=\left(\begin{smallmatrix}
	\hat{B} & O \\
	O & \check{B}
	\end{smallmatrix}\right)$, $\hat{B}\in\R^{(\hat{s}-\hat{n})\times(\hat{s}-\hat{n})}$, $\check{B}\in\R^{(\check{s}-\check{n})\times(\check{s}-\check{n})}$
	and $C=\left(\begin{smallmatrix}
	\hat{C} & O \\
	O & \check{C}
	\end{smallmatrix}\right)$, $\hat{C}\in\Z^{\hat{s}\times\hat{s}}$, $\check{C}\in\Z^{\check{s}\times\check{s}}$.
	Moreover, $\sigma(\hat{C})=\sigma(\hat{A})\cup\sigma(\hat{B})$, $\sigma(\check{C})=\sigma(\check{A})\cup\sigma(\check{B})$, and by assumption,
	no eigenvalue of $\hat{C}$ is an algebraic conjugate of an eigenvalue of $\check{C}$.
	
Write the matrix $L$ as a block matrix in the form
\begin{equation}\label{eq:velkaL}
L=
\begin{pmatrix}
\hat{L}_1 &\hat{J}_1\\
\check{J}_1 &\check{L}_1\\
\hat{L}_2 &\hat{J}_2\\
\check{J}_2 &\check{L}_2
\end{pmatrix}
\end{equation}
where the blocks are of dimensions suitable so that the following multiplication can be performed block-wise,
\begin{equation}\label{eq:cartesian}
  \begin{pmatrix}
	\hat{A} & O & O & O \\
	O & \check{A} & O & O \\
	O & O & \hat{B} & O \\
	O & O & O & \check{B}
	\end{pmatrix} \begin{pmatrix}
\hat{L}_1 &\hat{J}_1\\
\check{J}_1 &\check{L}_1\\
\hat{L}_2 &\hat{J}_2\\
\check{J}_2 &\check{L}_2
\end{pmatrix} =
	\begin{pmatrix}
\hat{L}_1 &\hat{J}_1\\
\check{J}_1 &\check{L}_1\\
\hat{L}_2 &\hat{J}_2\\
\check{J}_2 &\check{L}_2
\end{pmatrix}
	\begin{pmatrix}
	\hat{C} & O \\
	O & \check{C}
	\end{pmatrix}.
\end{equation}

We will now use a simple statement on matrices whose proof can be found for example in~\cite[Chapter 5]{Cullen}.

{\it Let $M,N$ be square matrices (not necessarily of the same order) such that their spectra are disjoint. Let further $X$ be a matrix such that $MX=XN$.
Then $X$ is a zero matrix.}

We apply the statement to~\eqref{eq:cartesian} wherefrom we derive that
$$
\begin{aligned}
\hat{A}\hat{J}_1 &= \hat{J}_1 \check{C},\\
\hat{B}\hat{J}_2 &= \hat{J}_2 \check{C},
\end{aligned}
\qquad
\begin{aligned}
\check{A}\check{J}_1 &= \check{J}_1 \hat{C},\\
\check{B}\check{J}_2 &= \check{J}_2 \hat{C},
\end{aligned}
$$
which implies that the matrices $\hat{J}_1$, $\hat{J}_2$, $\check{J}_1$, $\check{J}_2$ all vanish. We conclude that the cut-and-project scheme $\Lambda$ with the lattice $\L$ determined by the associated matrix $L$ in the form~\eqref{eq:velkaL} is by Definition~\ref{de:directsum} a direct sum of cut-and-project schemes with matrices $\hat{L}=\binom{\hat{L}_1}{\hat{L}_2}$, $\check{L}=\binom{\check{L}_1}{\check{L}_2}$. By Lemma~\ref{l:directsumgeneric}, the resulting scheme is generic.
\end{proof}

\section{Reformulation using inverse matrices}\label{sec:reformulationY}

Let us have a different look on equation~\eqref{eq:ABC} in Definition~\ref{de:self}. If $A\in\R^{n\times n}$ is a self-similarity of a given cut-and-project scheme $(\L\subset\R^s,\R^n)$ with lattice $\L$ determined by an associated matrix $L$ and if $B\in\R^{(s-n)\times(s-n)}$ and $C\in\Z^{s\times s}$ are matrices satisfying~\eqref{eq:ABC}, we obtain for the inverse $Y=L^{-1}$ of the matrix $L$ that
\begin{equation}\label{eq:ABC-Y}
  Y\begin{pmatrix}
     A & O \\
     O & B
   \end{pmatrix}=CY.
\end{equation}
Dividing the matrix $Y$ into two rectangular blocks $Y=(Y_1,Y_2)$, $Y_1\in\R^{s\times n}$, $Y_2\in\R^{s\times (s-n)}$, we derive from~\eqref{eq:ABC-Y} that $Y_1A=CY_1$ and $Y_2B=CY_2$. In other words, for $i=1,2$, the $\R$-span of columns of the matrix $Y_i$ is a real invariant subspace of the matrix $C$.
This is the reason why sometimes it is more convenient to work with the matrix $Y$ instead of the matrix $L$. Let us reformulate the necessary and sufficient conditions on the properties of a cut-and-project scheme.

\begin{prop}\label{p:generic}
	Let $\Lambda = (\L\subset \R^s, \R^n)$ be a cut-and-project scheme. Let $L \in \R^{s\times s}$ be a matrix associated to $\L$ and denote $Y = L^{-1}$. Then $\Lambda$ is
	\begin{itemize}
		\item[(i)] non-degenerate if and only if $\mathrm{span}_\C\{Y_{\bullet n+1},\dots,Y_{\bullet s}\}\cap\Q^s=\{\bs{0}\}$,
		\item[(ii)] aperiodic if and only if $\mathrm{span}_\C\{Y_{\bullet 1},\dots,Y_{\bullet n}\}\cap\Q^s=\{\bs{0}\}$,
		\item[(iii)] irreducible if and only if there exists $\bs{y} \in \mathrm{span}_\C\{Y_{\bullet 1},\dots,Y_{\bullet n}\}$ such that for all $\bs{q} \in \Z^s$, $\bs{q}\neq \bs{0}$ it holds that $\bs{q}^\top \bs{y} \neq 0$.
	\end{itemize}
\end{prop}

\begin{proof}
For the sake of the proof, write $Y=(Y_1,Y_2)$, where $Y_1\in\R^{s\times n}$, $Y_2\in\R^{s\times (s-n)}$.
First we show the statement in Item (i). Let $\bs{r}$ be a rational vector in $\mathrm{span}_\C\{Y_{\bullet n+1},\dots,Y_{\bullet s}\}$. Without loss of generality, consider $\bs{r}\in\Z^s$.  Then there exists $\bs{x}\in\C^{s-n}$
$$
\bs{r}=Y_2\bs{x}=Y\binom{O}{I_{s-n}}\bs{x}.
$$
Moreover, we have
	$$
	\pi_\parallel (L\bs{r})=(I_n,O) L\bs{r} = (I_n,O) LY\binom{O}{I_{s-n}}\bs{x} = (I_n,O) I_s \binom{O}{I_{s-n}}\bs{x} = \bs{0}\in\C^n\,.
	$$
Since by assumption of non-degeneracy $\pi_\parallel$ restricted to $\L$ is injective, and $L\bs{r}\in\L$, we obtain that $L\bs{r}=\bs{0}$. Consequently, $\bs{r}=\bs{0}$.

Assume on the other hand that the scheme is degenerate, i.e.\ there is a non-zero lattice vector $\bs{l}=L\bs{r}$ for some $\bs{0}\neq\bs{r}\in\Z^s$ such that
$\bs{0}=\pi_\parallel(L\bs{r})=(I_n,O) L\bs{r}$. Then
$$
L\bs{r}=\binom{\bs{0}}{\bs{x}} = \binom{O}{I_{s-n}}\bs{x}\quad \text{ for some } \bs{x}\in\C^{s-n}.
$$
This implies $\bs{r}=L^{-1}{I_{s-n}}\bs{x}=Y{I_{s-n}}\bs{x}$, i.e.\ $\bs{r}\in\mathrm{span}_\C\{Y_{\bullet n+1},\dots,Y_{\bullet s}\}\cap\Q^s$.
Analogically, we demonstrate the claim in Item (ii).
	
Let us focus on Item (iii). We will show it using Item (iii) of Proposition~\ref{p:pleasants}. Since $L$ is a real matrix, the vectors $\bs{u}$, $\bs{x}$ in Item (iii) of Proposition~\ref{p:pleasants} can be equivalently taken to be complex.
We therefore need to show equivalence of the following two statements:

(a) There exists $S\in\Z^{s\times(s-1)}$ such that for every $\bs{x}\in\C^n$ we have $\binom{\bs{x}}{\bs{0}}\in \{LS \bs{u}: \bs{u}\in\C^{s-1}\}$.
	
(b) For every $\bs{y} \in \mathrm{span}_\C\{Y_{\bullet 1},\dots,Y_{\bullet n}\}$ there exists $\bs{q} \in \Z^s$, $\bs{q}\neq \bs{0}$ such that $\bs{q}^\top \bs{y} = 0$.

Let us prove the implication (a)$\Rightarrow$(b). The fact that for every $\bs{x}\in\C^n$ we have $\binom{\bs{x}}{\bs{0}}\in \{LS \bs{u}: \bs{u}\in\C^{s-1}\}$ can be rewritten as
that for  every $\bs{x}\in\C^n$ we have $Y\binom{\bs{x}}{\bs{0}}\in \{S \bs{u}: \bs{u}\in\C^{s-1}\}$. This is further equivalent to saying that
$$
\mathrm{span}_\C\{Y_{\bullet 1},\dots,Y_{\bullet n}\}\subset \{S \bs{u}: \bs{u}\in\C^{s-1}\}.
$$
Consider a non-zero $\bs{q}\in\Z^s$ such that $\bs{q}^\top S = \bs{0}^\top\in\C^{1\times(s-1)}$. Since every $\bs{y}\in \mathrm{span}_\C\{Y_{\bullet 1},\dots,Y_{\bullet n}\}$ is of the form $\bs{y}=S\bs{u}$, $\bs{u}\in\C^{s-1}$, we obtain $\bs{q}^\top \bs{y} =\bs{q}^\top S\bs{u}=\bs{0}^\top\bs{u} = 0$.
	
In order to prove (b)$\Rightarrow$(a), consider the vector space
$$
{\mathcal E}= \big\{\sum_{i=1}^s\sum_{j=1}^{n}q_{ij}Y_{ij} : q_{ij}\in\Q\big\}\subset\R
$$
over the rationals. The space ${\mathcal E}$ is generated by coordinates of the columns of the matrix $Y_1$. Denote $d=\dim_{\Q}{\mathcal E}$ and choose a basis $\gamma_1,\dots,\gamma_d$ of ${\mathcal E}$. Since the dimension of $\R$ as a vector space over $\Q$ is infinite, there exist real numbers $t_1,\dots,t_n\in\R$ such that $t_i\gamma_j$, $i=1,\dots,n$, $j=1,\dots,d$, are linearly independent over $\Q$. Choose
$$
\bs{y}=t_1 Y_{\bullet 1} +t_2 Y_{\bullet 2} + \cdots + t_n Y_{\bullet n}.
$$
Statement (b) implies existence of a non-zero $\bs{q}\in\Z^s$ such that
$$
0=\bs{q}^\top \bs{y} = t_1 \bs{q}^\top Y_{\bullet 1} +t_2 \bs{q}^\top Y_{\bullet 2} + \cdots + t_n \bs{q}^\top Y_{\bullet n}.
$$
As $\bs{q}^\top Y_{\bullet j}$ belongs to $\mathcal{E} = \mathrm{span}_\Q\{\gamma_1,\dots,\gamma_d\}$, we can write for
every $i =1,2,...,n$ that $\bs{q}^\top Y_{\bullet i} = \sum_{j=1}^d q_{ij}\gamma_j$. Then
$0 = \sum_{i=1}^n  t_i \bs{q}^\top Y_{\bullet i} = \sum_{i=1}^{n}\sum_{j=1}^{s}  q_{ij} t_i\gamma_j$.
Since numbers $t_i\gamma_j$ are linearly independent over $\Q$, the latter implies that $q_{ij}=0$, and thus $\bs{q}^\top Y_{\bullet i}=0$
 for every $i=1,\dots,n$. Denoting $\bs{q}^\top =(q_1,q_2,\dots,q_s)$, this can be rewritten as
\begin{equation}\label{eq:zs}
  q_s Y_{s i} = -q_1 Y_{1 i} - q_2 Y_{2 i} - \cdots -q_{s-1} Y_{(s-1) i}.
\end{equation}

Since $\bs{q}\neq \bs{0}$, we can assume without loss of generality that $q_s\neq 0$.
Set
$$
S=\begin{pmatrix}
q_s&0&\cdots&0\\
0 &q_s&\cdots&0\\
\vdots&\vdots&\ddots&\vdots\\
0&0&\cdots&q_s\\
-q_1&-q_2&\cdots&-q_{s\!-\!1}
\end{pmatrix}\in\Z^{s\times (s-1)}
$$
and for $i=1,2,\dots,n$
$$
\bs{u_i} =
\begin{pmatrix}
Y_{1i}\\
Y_{2i}\\
\vdots\\
Y_{(s-1)i}
\end{pmatrix}.
$$
Then  using~\eqref{eq:zs}, we have $S\bs{u}_i = q_sY_{\bullet i}$.

Now let us show that for every $\bs{x}\in\C^n$ the vector $\binom{\bs{x}}{\bs{0}}\in\C^s$ can be written as $LS\bs{u}$ for some $\bs{u}\in\C^{s-1}$.
Consider $\bs{x}\in\C^n$, $\bs{x}^\top = (x_1,x_2,\dots,x_n)$. Set
$$
\bs{u}^\top = \frac{1}{q_s}(x_1\bs{u}_1+x_2\bs{u}_2+\cdots+x_n\bs{u}_n)^\top.
$$
Then, denoting $\bs{e}_1,\dots,\bs{e}_s$ the standard basis of $\R^s$, and since $L$ is the inverse to $Y=(Y_1,Y_2)$, we have
$$
LS\bs{u} = L (x_1 Y_{\bullet 1} +x_2 Y_{\bullet 2} + \cdots + x_n Y_{\bullet n}) =
x_1\bs{e}_1 + x_2\bs{e}_2+\cdots+x_n\bs{e}_n
=
\binom{\bs{x}}{\bs{0}}.
$$
This completes the proof of Item (iii).
\end{proof}


\section{Vandermonde cut-and-project schemes}\label{sec:element}


In this section we demonstrate the construction of an elementary scheme corresponding to a monic polynomial $f\in\Z[x]$ irreducible over $\Q$. 
The construction will provide a generic cut-and-project scheme if $f$ is not linear or quadratic with negative discriminant.
The lattice $\L$ is constructed using the inverse to the Vandermonde matrix of $f$ (or its real version).
Denote the mutually distinct roots of $f$ by $\beta_j$, $j=1,\dots,d$. The Vandermonde matrix $Z_f$ of the polynomial $f$ is of the form
$$
Z=Z_f=\begin{pmatrix}
1&1&\cdots&1\\
\beta_1 &\beta_2&\cdots&\beta_d\\[1mm]
\beta_1^2&\beta_2^2&\cdots&\beta_d^2\\
\vdots&\vdots&\cdots&\vdots\\[1mm]
\beta_1^{d-1}&\beta_2^{d-1}&\cdots&\beta_d^{d-1}
\end{pmatrix}.
$$
The Vandermonde matrix is non-singular, since numbers $\beta_j$ are mutually distinct.
Note that $\beta_j$ are eigenvalues of the companion matrix $C=C_f$ and the corresponding eigenvectors can be chosen to be the columns of the matrix $Z_f$, i.e.\  $\bs{z}_j=(1,\beta_j,\beta_j^2,\dots,\beta_j^{d-1})^\top$. This implies
$$
CZ=ZD\,\quad \text{ where } D=\,\mathrm{diag}(\beta_1,\dots,\beta_d)\,.
$$
Any choice of the columns of $Z$ therefore generates an invariant subspace of $C$. We will further need the following fact about the matrix $Z=Z_f$.

\begin{lem}\label{l:obrazy}
  Let $\beta_1,\dots,\beta_d$ be roots of a polynomial $f\in\Z[x]$ irreducible over $\Q$.
  Let $Z$ be the $d\times d$ matrix defined above. Then the $j$-th row of the matrix $Z^{-1}$ has components in the field $\Q(\beta_j)$ and is the image of the first row of $Z^{-1}$ under the field isomorphism $\psi:\Q(\beta_1)\to\Q(\beta_j)$ defined by the assignment $\beta_1\mapsto\beta_j$.
\end{lem}

\begin{proof}
 Denote for simplicity by $\bs{w}_j^\top$ the $j$-th row $(Z^{-1})_{j\bullet}$ of the matrix $Z^{-1}$. The vector $\bs{w}_j$ is uniquely given by the relations
 $\bs{w}_j^\top\bs{z}_k=\delta_{jk}$, $k=1,\dots,d$ where $\delta_{jk}$ stands for the Kronecker symbol.

  Recall that $CZ=ZD$. Hence $Z^{-1}C=DZ^{-1}$ and after transposition
  $ C^\top(Z^{-1})^\top = (Z^{-1})^\top D$. Obviously, for any $i$, the $i$-th column $\bs{w}_i$ of $(Z^{-1})^\top$ is an eigenvector of the matrix $C^\top$ corresponding to the eigenvalue $\beta_i$, i.e.\ $C^\top\bs{w}_i = \beta_i\bs{w}_i$. Its components therefore belong to $c\cdot\Q(\beta_i)$ for a constant $c$.  As $\bs{w}_i^\top\bs{z}_i=1$, we derive that  $\bs{w}_i^\top\in\Q(\beta_i)$. Applying the field isomorphism $\psi$ on the equation $C^\top\bs{w}_1 = \beta_1\bs{w}_1$, we obtain (with a little abuse of notation)
  that
  $$
  C^\top\psi(\bs{w}_1) = \beta_j\psi(\bs{w}_1)\,,
  $$
  where we have used the fact that $\psi$ restricted to $\Q$ is the identity. We have derived that $\psi(\bs{w}_1)\in\Q(\beta_j)$ is an eigenvector of $C^\top$ corresponding to the eigenvalue $\beta_j$, i.e.\ $\psi(\bs{w}_1)=t\bs{w}_j$ for a constant $t\in\Q(\beta_j)$. Knowing that $\psi(\bs{z}_1)=\bs{z}_j$, we further have $t=t\bs{w}_j^\top\bs{z}_j=\psi(\bs{w}_1^\top\bs{z}_1)=\psi(1)=1$. Thus  $\bs{w}_j=\psi(\bs{w}_1)$.
\end{proof}

The above lemma has a simple corollary. Having a rational vector $\bs{r}\in\Q^d$, the components of the vector $Z^{-1}\bs{r}$ are images under Galois automorphisms in the splitting field $\F=\Q(\beta_1,\dots,\beta_d)$ of the polynomial $f$. Therefore the components must be all non-zero or all equal to 0.

\begin{coro}\label{c:vektor0}
  Let $Z$ be as above and let $\bs{r}\in\Q^{d}$. If at least one component of the vector $Z^{-1}\bs{r}$ vanishes, then $\bs{r}=\bs{0}$.
\end{coro}

Since all eigenvalues of $C=C_f$ have multiplicity 1, with the eigenvectors $\bs{z}_j$, one can form two kinds of elementary real invariant subspaces of $C$, namely
$$
\begin{aligned}
\mathrm{span}_\R \{\bs{z}_j\}\quad &\text{if }\beta_j\in\R\\
\mathrm{span}_\R \{\Re\bs{z}_j,\Im\bs{z}_j\}\quad &\text{if }\beta_j\in\C\setminus\R,
\end{aligned}
$$
where the real part $\Re$, and the imaginary part $\Im$ are taken componentwise.
It is obvious that every real invariant subspace of $C$ is a direct sum of a selection of the above elementary subspaces.

Consider a decomposition of $\R^d$ into a direct sum of two real invariant subspaces $\mathcal{Y}_1,\mathcal{Y}_2$ of the matrix $C$ of dimension $1\leq n,d-n < d$, respectively, $\R^d=\mathcal{Y}_1\oplus\mathcal{Y}_2$.  Necessarily, each of the subspaces must have a basis formed by vectors $\bs{z}_j$, or pairs $\Re\bs{z}_j,\Im\bs{z}_j$ for some indices $j$. Take matrices $Y_1$, $Y_2$ formed from the column vectors of these bases, respectively.
 Now define the matrix $Y=(Y_1,Y_2)$ and $L=Y^{-1}$. We denote further $L_1\in\R^{n\times d}$, $L_2\in\R^{(d-n)\times d}$ such that $L=\binom{L_1}{L_2}$. We call the cut-and-project scheme $\Lambda=(\L\subset\R^d,\R^n)$, where $\L$ has $L$ as its associated matrix, \emph{a Vandermonde cut-and-project scheme corresponding to the polynomial $f$}. 
 
 \begin{pozn}\label{pozn:Vandermonde}
  Note that the lattice $\L$ constructed in a Vandermonde scheme is given by the matrix $L=Y^{-1}$ which has coordinates in the splitting field $\F=\Q(\beta_1,\dots,\beta_d)$ of the polynomial $f$.
  To a polynomial $f$, one can construct different Vandermonde schemes, even if the dimension $n$ is fixed. Such two schemes differ by permutation of columns in the matrix $Y$, or in other words, by permutation of rows in the matrix $L$.
  \end{pozn}

 In order to prove that the cut-and-project scheme constructed in this way is generic, we state two simple lemmas.

\begin{lem}\label{l:elem1}
  Let $\mathcal{I}$ be a proper non-empty subset of $\{1,\dots,d\}$. Then we have $\mathrm{span}_\C\{\bs{z}_j: j\in{\cal I}\}\cap \Q^{d}=\{\bs{0}\}$.
\end{lem}

\begin{proof}
  Suppose that for some $\alpha_j\in\C$, we have $\sum_{j\in\mathcal{I}}\alpha_j\bs{z}_j=\bs{r}\in\Q^d$. Applying the matrix $Z^{-1}$ we obtain
  $$
  Z^{-1}\sum_{j\in\mathcal{I}}\alpha_j\bs{z}_j=\bs{v}=Z^{-1}\bs{r},
  $$
  where $\bs{v}=(v_1,\dots,v_d)^\top$ with $v_j=\alpha_j$ if $j\in\mathcal{I}$ and $v_j=0$ otherwise. Note that $v_j=0$ for at least one index $j$, since $\mathcal{I}$ is a proper subset of $\{1,\dots,d\}$. By Corollary~\ref{c:vektor0}, $\bs{r}=\bs{0}$.
\end{proof}

\begin{lem}\label{l:A}
 Let $\bs{z}_j$ be the $j$th column of the Vandermonde matrix $Z$. Then for every  $\bs{q}\in\Z^d$, $\bs{q}\neq\bs{0}$, we have $\bs{q}^\top\bs{z}_j\neq 0$.
\end{lem}

\begin{proof}
  Let $\bs{q}\in\Z^d$ and write $\bs{q}^\top=(q_0,q_1,\dots,q_{d-1})$. Then
  $$
  \bs{q}^\top\bs{z}_j = q_0+q_1\beta_j + q_2\beta_j^2 + \cdots + q_{d-1}\beta_j^{d-1}=0
  $$
  means that $\beta_j$ is a root of an integer polynomial of degree strictly smaller than $d$. If $\bs{q}\neq \bs{0}$, then we have a contradiction to the fact that $\beta_j$ is an algebraic integer of degree $d$.
\end{proof}

\begin{thm}\label{t:elementary}
Let $f\in\Z[X]$ be a monic polynomial irreducible over $\Q$. Any Vandermonde cut-and-project scheme $\Lambda$ corresponding to $f$ is generic.
\end{thm}

\begin{proof}
Lemma~\ref{l:elem1} implies by Proposition~\ref{p:generic} that the cut-and-project scheme $(\L\subset\R^d,\R^n)$ in Theorem~\ref{t:elementary} is aperiodic. For that, it suffices to realize that there exists a set of indices $\mathcal{I}$ (a proper non-empty subset of $\{1,\dots,d\}$) such that $\mathrm{span}_\C\{ Y_{\bullet 1},\dots,Y_{\bullet n}  \}= \mathrm{span}_\C\{\bs{z}_j: j\in{\cal I}\}$.
In a similar way, one proves non-degeneracy.
In order to prove irreducibility, we use criterion (iii) in
Proposition~\ref{p:generic} combined with Lemma~\ref{l:A}.
\end{proof}


\section{Construction for trivial self-similarities}\label{sec:trivial}

In this section we present the construction of a cut-and-project scheme with self-similarities $A$ which have in the spectrum rational integers or non-real quadratic numbers.
These are exactly the cases in which the Vandermonde scheme cannot be found.

If $A=kI$ for some $k\in\Z$, we have $\sigma(A) = \{k\} \subset\Z$, and it is obvious that any generic scheme $(\L\subset\R^{n+1},\R^n)$ satisfies the required properties. For, any $\Z$-module is closed under multiplication by integers.

Suppose now that the spectrum $\sigma(A)$ contains only complex conjugated pairs of the same  non-real quadratic numbers $\lambda, \bar{\lambda}$.

\begin{prop}
	Let $A\in \R^{n\times n}$ be a matrix diagonalizable over $\C$ such that its eigenvalues are roots of the same quadratic polynomial $\lambda^2 -p\lambda - q$, $p,q\in \Z$ with negative discriminant. Then there exists a generic scheme $\Lambda=(\L \subset \R^{n+2},\R^n)$ with self-similarity $A$ and $n+2$ is the minimal dimension of a lattice defining such a scheme.
\end{prop}

\begin{proof}
	It is obvious that $n$ is even, i.e.\ $n=2m$ for some $m$. Due to Corollary~\ref{c:BUNO}, the matrix $A$ can be assumed to be in the form
$$
A = I_m \otimes \begin{pmatrix}
	\Re \lambda & -\Im \lambda \\ \Im \lambda & \Re \lambda
	\end{pmatrix}.
$$
Let us define the lattice $\L$ through the following matrix
$$
L = H \otimes \begin{pmatrix}
	1 & \Re \lambda \\ 0 & \Im \lambda
	\end{pmatrix}\qquad\text{ for some matrix } H\in \R^{(m+1)\times(m+1)}.
$$
The conditions on the matrix $H$ will be imposed later. At the same time let us define $C  = I_{m+1} \otimes \left(\begin{smallmatrix}
	0 & q \\ 1 & p
	\end{smallmatrix}\right)$ and $B = \left(\begin{smallmatrix}
	\Re \lambda & -\Im \lambda \\ \Im \lambda & \Re \lambda
	\end{smallmatrix} \right)$. We verify that equation \eqref{eq:ABC} holds. We want to have
$$	
\begin{gathered}
\left(H \otimes  \begin{pmatrix}
	1 & \Re \lambda \\ 0 & \Im \lambda
	\end{pmatrix}\right) \left( I_{m+1} \otimes \begin{pmatrix}
	0 & q \\ 1 & p
	\end{pmatrix} \right) = \hspace*{5cm}\\
\hspace*{4cm}=\left(I_{m+1} \otimes \begin{pmatrix}
	\Re \lambda & -\Im \lambda \\ \Im \lambda & \Re \lambda
	\end{pmatrix}\right)\left(H \otimes  \begin{pmatrix}
	1 & \Re \lambda \\ 0 & \Im \lambda
	\end{pmatrix} \right),
\end{gathered}
$$
which is equivalent to 
$$
H \otimes \begin{pmatrix}
	\Re \lambda & q + p\, \Re \lambda \\ \Im \lambda & p\, \Im \lambda
	\end{pmatrix} = H \otimes \begin{pmatrix}
	\Re \lambda & \Re^2 \lambda - \Im ^2 \lambda \\ \Im \lambda & 2 \Im \lambda\, \Re \lambda
	\end{pmatrix}.
$$
The latter equality is true, because of $\Re^2 \lambda - \Im ^2 \lambda =  q + p\,\Re \lambda$ and $2 \Im \lambda\, \Re \lambda = p \Im \lambda$. This proves by Proposition \ref{p:ABCzpet} that $A\pi_\parallel(\L) \subset \pi_\parallel(\L)$.

Let us now prove that under suitable assumptions on the matrix $H$, we obtain a generic cut-and-project scheme $\Lambda=(\L\subset\R^{n+2},\R^n)$.
The non-degeneracy of $\Lambda$ is ensured choosing numbers $H_{m\, i}$, $i = 1, 2 ,\dots ,m+1$, linearly independent over $\Q$. Indeed, if $\bs{l}$ is a lattice point, then its image under the projection $\pi_\parallel$ is of the form
	\[
    \pi_\parallel(\bs{l}) = \sum_{i=1}^{m+1} a_i \begin{pmatrix}
	H_ {1\, i} \\ 0 \\ \vdots \\  H_ {m \, i} \\ 0
	\end{pmatrix} + \sum_{i=1}^{m+1} b_i \begin{pmatrix}
	H_{1\, i} \, \Re \lambda \\ H_{1\, i} \, \Im \lambda \\ \vdots \\ H_{m\, i} \, \Re \lambda \\ H_{m\, i} \, \Im \lambda
	\end{pmatrix}, \ a_i,b_i \in \Z.
   \]
	If $\pi_\parallel(\bs{l})=\bs{0}$, then the last row of the above must vanish. More precisely,
	\[
 \Im \lambda \sum_{i=1}^{m+1} b_i  H_{m\, i}  =0.
  \]
	From the requirement on the linear independence of  $H_{m\, i}$ over $\Q$ one gets $b_i = 0 $ for $i = 1,\dots , m+1$. Consequently, from the penultimate row, we check that $a_i =0$ for $i = 1,\dots, m+1$. Necessarily, $\bs{l}=\bs{0}$ and hence $\pi_\parallel$ restricted to $\L$ is injective.

In order to prove aperiodicity of the scheme, consider the $\pi_\perp$ image of a lattice point, which is of the form
	\[
    \pi_\perp(\bs{l}) = \sum_{i=1}^{m+1} a_i H_ {m+1,\, i}
      \begin{pmatrix}
	1 \\ 0
	\end{pmatrix} + \sum_{i=1}^{m+1} b_i H_{m+1,\, i}
    \begin{pmatrix}
	 \Re \lambda \\ \Im \lambda
	\end{pmatrix} : a_i,b_i \in \Z.
   \]
Similarly as before, choosing numbers $H_{m+1,\, i}$,  $i = 1, 2 ,\dots ,m+1$, to be linearly independent over $\Q$, we see that $\pi_\perp(\bs{l})$ implies $\bs{l}=\bs{0}$ and thus $\pi_\perp$ restricted to $\L$ is injective.

It remains to verify that the cut-and-project scheme $\Lambda$ is irreducible.
The projection in the internal space $\R^2$ can be expressed as
\begin{equation}\label{eq:trivialH}
   \pi_\perp(\L) = \Big\{
   \sum_{i=1}^{m+1} a_i H_ {m+1,\, i}
      \begin{pmatrix}
	1 \\ 0
	\end{pmatrix} + \sum_{i=1}^{m+1} b_i H_{m+1,\, i}
    \begin{pmatrix}
	 \Re \lambda \\ \Im \lambda
	\end{pmatrix} : a_i,b_i \in \Z\Big\} .
	\end{equation}
Consider the set $ \mathcal{H} = \left\{\sum_{i=1}^{m+1}c_i H_{m+1,\, i}  :c_1,\dots,c_{m+1} \in \Z \right\}$.
The second row of the projection~\eqref{eq:trivialH} can be expressed as $({\Im \lambda}) \mathcal{H}$. Since $H_{m+1,\, i}$,  $i = 1, 2 ,\dots ,m+1$, are linearly independent over $\Q$, the set $\mathcal{H}$ is dense in $\R$. The first row of the projection $\pi_\perp(\L)$ is then written as $\mathcal{H} +({\Re \lambda}) \mathcal{H}$. This set is clearly dense in $\R$ as well. From~\eqref{eq:trivialH}, we conclude that $\pi_\perp(\L)$ is dense in $\R^2$. We conclude that the scheme $\Lambda$  is generic and has the matrix $A$ as its self-similarity.
\end{proof}

\section{Answer to Question~\ref{q2}}\label{sec:answerQ2}

In this section we prove the key result for answering Question~\ref{q2}.

\begin{thm}\label{t:mindimensionnotmin}
Let $A\in\R^{n\times n}$ be a matrix diagonalizable over $\C$.
Then there exists a generic cut-and-project scheme $\Lambda$
with self-similarity $A$ if and only if the spectrum $\sigma(A)$ of $A$ is composed of
algebraic integers.
\end{thm}

\begin{proof}
Necessity is obvious from Corollary~\ref{c:duseldek1}. For the opposite implication, we must show existence of a generic cut-and-project scheme with self-similarity $A$.
Based on Proposition~\ref{p:Anekonjug}, we may consider, without loss of generality, only matrices $A$ whose eigenvalues are roots of the same minimal polynomial $f$. If $f$ is linear or quadratic with negative discriminant, the situation is solved in Section~\ref{sec:trivial}. Suppose that $f$ is quadratic with positive discriminant or of degree at least three.
 To every $\beta$ in the spectrum of $A$ we create a Vandermonde cut-and-project scheme $\Lambda_\beta$ corresponding to the minimal polynomial $f$ of $\beta$. We take the scheme $\Lambda_\beta$ so that it has self-similarity $A_\beta=(\beta)\in\R^{1\times 1}$ if $\beta\in\R$ or
$A_\beta=\left(\begin{smallmatrix}
\Re \beta & -\Im\beta \\
\Im\beta & \Re\beta
\end{smallmatrix}\right)\in\R^{2\times2}$
if $\beta\in\C\setminus\R$.
For every real eigenvalue $\beta$ or complex pair of eigenvalues $\beta,\overline{\beta}$ we compose as many schemes $\Lambda_\beta$ as is its multiplicity in the spectrum of the matrix $A$.
By Lemmas~\ref{l:directsumgeneric} and~\ref{l:sklad}, the direct sum of all these cut-and-project schemes is a generic cut-and-project scheme with self-similarity $A$.
\end{proof}

We illustrate the construction of the above proof on a simple example.

\begin{example}\label{ex:dimenze6}
  Consider the matrix
  $$
  A=\begin{pmatrix}
                             \tau & 0 & 0 \\
                             0 & \tau & 0 \\
                             0 & 0 & \tau'
                           \end{pmatrix},
  $$
  where $\tau=\frac12(1+\sqrt5)$ is the golden ratio and $\tau'=\frac12(1-\sqrt5)$ its algebraic conjugate.
 The minimal polynomial of $\tau,\tau'$ is $f(x)=x^2-x-1$ with companion matrix $C_f=\left(\begin{smallmatrix}
                             0 & 1 \\
                             1 & 1
                           \end{smallmatrix}\right)$.
 The matrix $C_f$ has eigenvectors $\binom{1}{\tau}$, $\binom{1}{\tau'}$ corresponding to the eigenvalues $\tau$, $\tau'$, respectively.

The proof of Theorem~\ref{t:mindimensionnotmin} gives a construction of a cut-and-project scheme $\Lambda=(\L\subset\R^6,\R^3)$ with self-similarity $A\pi_\parallel(\L)\subset\pi_\parallel(\L)$  as a direct sum of three Vandermonde schemes to the polynomial $f$.  The resulting scheme $(\L\subset\R^6,\R^3)$ is given by the lattice $\L$ associated to the matrix $L$ of the form
  $$
    L=\begin{pmatrix}
        1 & 0 & 0 & 1 & 0 & 0 \\
        \tau & 0 & 0 & \tau' & 0 & 0 \\
        0 & 1 & 0 & 0 & 1 & 0 \\
        0 & \tau & 0 & 0 & \tau' & 0 \\
        0 & 0 & 1 & 0 & 0 & 1 \\
        0 & 0 & \tau' & 0 & 0 & \tau
      \end{pmatrix}^{-1}.
  $$
Matrices $B,C$ for~\eqref{eq:ABC} are
    $$
    C=\begin{pmatrix}
                             C_f & O & O \\
                             O & C_f & O \\
                             O & O & C_f
                           \end{pmatrix}\in\Z^{6\times6},\qquad
    B=\begin{pmatrix}
                             \tau' & 0 & 0 \\
                             0 & \tau' & 0 \\
                             0 & 0 & \tau
                           \end{pmatrix}\in\R^{3\times 3}.
    $$
\end{example}

In some cases, the construction in the proof of Theorem~\ref{t:mindimensionnotmin} gives a cut-and-project scheme with a lattice in an exaggeratedly high dimension. We demonstrate such a situation on the following example.

\begin{example}\label{ex:kubicke}
Let $f(x)=x^3-2x^2-x+1$. The polynomial has three real roots $\beta,\beta',\beta''$. Let us construct a cut-and-project scheme with self-similarity given by
$$
A=\begin{pmatrix}
    \beta & 0 &0 &0 \\
    0 & \beta & 0 & 0 \\
    0 & 0 & \beta' & 0 \\
    0 & 0 & 0 & \beta''
  \end{pmatrix}.
$$
The construction described in proof of Theorem~\ref{t:mindimensionnotmin} yields a cut-and-project scheme as a direct sum of four Vandermonde schemes, each of dimension $d=3$, thus giving a lattice $\L\subset\R^{12}$. Let us provide a construction using a direct sum of only two Vandermonde schemes, resulting a lattice $\L\in\R^6$.
Consider two Vandermonde schemes $\Lambda_1$, $\Lambda_2$ to the polynomial $f$. The scheme $\Lambda_1$ is constructed in such a way that it has a self-similarity
$A_1=\big(\begin{smallmatrix}
  \beta & 0 \\
    0 & \beta'
\end{smallmatrix}\big)$, and the matrix $B_1$ corresponding to $A_1$ is of the form $B_1=(\beta'')$. Similarly, the scheme $\Lambda_2$ is constructed to have a self-similarity
$A_2=\big(\begin{smallmatrix}
  \beta & 0 \\
    0 & \beta''
\end{smallmatrix}\big)$, and the matrix $B_2$ corresponding to $A_2$ is of the form $B_2=(\beta')$.
The cut-and-project scheme $\Lambda$ having as a self-similarity the matrix $A_1\oplus A_2$ is constructed as the direct sum $\Lambda_1\oplus\Lambda_2$. The scheme with self-similarity $A\sim_\R A_1\oplus A_2$ is obtained by a simple permutation matrix.

Note that in order to obtain the schemes $\Lambda_1$, $\Lambda_2$, we had to distribute the eigenvalues of $A$ into the spectrum of $A_1$, $A_2$ in a suitable way, so that the corresponding matrices $B_1$, $B_2$ are non-empty. This can be schematically captured by providing a matrix ${\mathcal K}\in\{0,1\}^{d\times 2}$ with $d$ rows corresponding to the eigenvalues of $A$ and $2$ columns corresponding to the matrices $A_1,A_2$. An element $\K_{ij}$ is equal to 1, if the $i$-th eigenvalue is assigned into the spectrum of $A_j$,
$$
\K=\begin{pmatrix}
  1 & 1  \\
  1 & 0  \\
  0 & 1
\end{pmatrix}\,.
$$
The number of elements ``1'' in the $i$-th row is equal to the multiplicity of the $i$-th eigenvalue in $\sigma(A)$. Moreover, in order that the matrices $A_1,A_2,B_1,B_2$ are all non-empty, each column must have at least one ``1'' and at least one ``0''.
\end{example}

As illustrated on the above example, the construction of a cut-and-project scheme to a matrix $A$ by composition of elementary Vandermonde schemes can be reformulated into solving a combinatorial question of existence of a certain type of matrix. Let $u,K\in \N$ and let $l_1,\dots,l_u$ be non-negative integers.
We will say that a matrix ${\mathcal K}\in\{0,1\}^{u\times K}$ is a well distributing matrix if it satisfies
$$
(i) \quad \sum_{j=1}^{K}{\mathcal K}_{i,j}=l_i,\ \text{ for all }i,\qquad (ii) \quad 1\leq \sum_{i=1}^{u}{\mathcal K}_{i,j}\leq u-1, \ \text{ for all }j.
$$

\begin{pozn}
  The combinatorial question could also be rephrased as a problem of construction of a bipartite graph without multiple edges. Consider a graph with $u$ vertices in the first and $K$ vertices in the second part of the graph. Elements ${\mathcal K}_{i,j}$ determine whether vertex $i$ in the first and vertex $j$ in the second part are connected by an edge or not.
  Condition (i) expresses the fact that the degrees of vertices in the first part are equal to $l_1,\dots,l_d$, condition (ii) states that the degree of each vertex in the second part takes value in $\{1,2,\dots,u-1\}$. The existence of such a graph is equivalent to existence of a matrix ${\mathcal K}\in\{0,1\}^{u\times K}$ with the desired conditions satisfied.
\end{pozn}

\begin{lem}\label{l:K}
Let $u\in\N$, $u\geq 2$, and let $l_1,\dots,l_u$ be non-negative integers. Denote $M:=\max\{l_j:j=1,\dots,u\}$. The minimal value $K$ for which a well distributing matrix $\K\in\{0,1\}^{u\times K}$ exists is
\begin{equation}\label{eq:K}
\max\Big\{M, \Bigl\lceil\tfrac{1}{u-1}\sum_{i=1}^{u}l_i\Bigr\rceil\Big\}.
\end{equation}
\end{lem}

\begin{proof}
Let us show that $K$ cannot be smaller than the value in~\eqref{eq:K}.
Consider a well distributing matrix $\K\in\{0,1\}^{u\times K}$.
Condition (i) says that the $i$-th row of the matrix ${\mathcal K}$ contains $l_i$ elements ``1'' and $K-l_i$ elements ``0''.
Necessarily, $K\geq M=\max\{l_j:j=1,\dots,u\}$.
A well distributing matrix ${\mathcal{K}}$ also satisfies (ii). Summing inequalities (ii) for $j=1,\dots, K$, we have
$$
\sum_{j=1}^K\sum_{i=1}^{u}{\mathcal K}_{i,j}=\sum_{i=1}^{u} l_i \leq K(u-1)\,,
$$
which shows that $K\geq \tfrac{1}{u-1}\sum_{i=1}^{u}l_i$.

Let us demonstrate that for $K:=\max\Big\{M, \Bigl\lceil\tfrac{1}{u-1}\sum_{i=1}^{u}l_i\Bigr\rceil\Big\}$, a well distributing matrix exists.
Suppose first that $l_1+l_2\leq K$.
We will define the matrix ${\mathcal K}$ by setting precisely $l_i$ elements ``1'' in the $i$-th row, thus respecting (i).
Define the first two rows of the matrix ${\mathcal K}$ as follows.
$$
{\mathcal K}_{1,j}=\begin{cases}
                     1 & \mbox{for } j=1,\dots, l_1 ,\\
                     0 & \mbox{otherwise},
                   \end{cases}
                   \qquad
{\mathcal K}_{2,j}=\begin{cases}
                     1 & \mbox{for } j=l_1+1,\dots, l_1+l_2, \\
                     0 & \mbox{otherwise}.
                   \end{cases}
$$
Note that whatever the elements in the other rows of ${\mathcal K}$, we already have $\sum_{i=1}^{u}{\mathcal K}_{i,j}\leq u-1$ for all~$j$.
If $u=2$, then~\eqref{eq:K} means that $K=l_1+l_2$,
$$
{\mathcal K}=\Big(\begin{array}{c@{\ \,}c}
               \overbrace{1\ 1\ \dots\ \dots\ 1}^{l_1} & 0\ 0\ \dots\ 0 \\
                0\ 0\ \dots\ \dots\ 0 & \underbrace{1\ 1\ \dots\ 1}_{l_2}
             \end{array}\Big)
$$
and we have also $\sum_{i=1}^{u}{\mathcal K}_{i,j}\geq 1$ for all $j$, thus ${\mathcal K}$ satisfies (ii).
Consider $u\geq 3$.
By assumption, $K\leq\sum_{i=1}^{u} l_i$, and therefore there exists $i_0\geq 2$ such that $\sum_{i=1}^{i_0}l_i\leq K<\sum_{i=1}^{i_0+1}l_i$.
Set for $i=3,\dots,i_0$
$$
{\mathcal K}_{i,j}=\begin{cases}
                     1 & \mbox{for } j=1+\sum_{m=1}^{i-1}l_m,\dots, \sum_{m=1}^{i}l_m \\
                     0 & \mbox{otherwise},
                   \end{cases}
$$
and ${\mathcal K}_{i_0+1,j}=1$ for $j=1+\sum_{m=1}^{i_0}l_m,\dots, K$. The rest of matrix elements define arbitrarily, just respecting (i), i.e.
$$
{\mathcal K}=\begin{pmatrix}
               \underbrace{1\ \dots\ 1}_{l_1} &  &  &  &  &    \\[-3mm]
                & \underbrace{1\ \dots\ 1}_{l_2} &  &  &  &   \\[-3mm]
                &  & \underbrace{1\ \dots\ 1}_{l_3} &  &  &   \\[-3mm]
                &  &  & \quad\dots\quad &  &   \\[-3mm]
                &  &  &  & \overbrace{1\ \dots\ 1}^{l_{i_0}} &   \\
               * & * & * & * & * & 1\ \dots  \\
               * & * & * & * & * & *  \\
               * & * & * & * & * & *
             \end{pmatrix}
$$
With this, $\sum_{i=1}^{u}{\mathcal K}_{i,j}\geq 1$ for all $j$. 

Secondly, let $l_1+l_2>K$. This is equivalent to $(K-l_1)+(K-l_2)<K$, and the demonstration is analogous to the above, just interchanging
the role of ``0'' and ``1'' as matrix elements and $K-l_i$ and $l_i$ as their number in the row.
Here we use the inequality $\frac{1}{u-1}\sum_{i=1}^{u}l_i\leq K$, which can be rewritten as $K\leq \sum_{i=1}^{u} (K-l_i)$.
This completes the proof.
\end{proof}

\begin{prop}\label{p:notmindimension}
Let $A\in\R^{n\times n}$ be a matrix diagonalizable over $\C$. Suppose that
all $\lambda\in\sigma(A)$ are algebraic integers with the same minimal polynomial $f$ with $r$ real and $t$ pairs of complex conjugate roots. Denote $l_j=m_A(\beta_j)$, $j=1,\dots,r+t$, and $M:=\max\{l_j:j\in\{1,\dots,r+t\}\}$. Suppose that $r+t\geq 2$.
Set
$$
s= d\max\Bigl\{ M, \Bigl\lceil\tfrac{1}{r+t-1}\sum_{j=1}^{r+t}l_j\Bigr\rceil\Bigr\}\,.
$$
Then there exists a scheme $\Lambda=(\L\subset\R^{s},\R^n)$ with self-similarity $A$.
\end{prop}

The assumption $r+t\geq 2$ in the statement is equivalent to the requirement that the eigenvalues of $A$ are neither integers nor non-real quadratic numbers. 
The minimal cut-and-project scheme with a self-similarity $A$ whose spectrum is of the form $\sigma(A)=\{k\}\subset\Z$, or $\sigma(A)=\{a\pm b i\}$ where $a+bi$ is a quadratic number has been constructed in Section~\ref{sec:trivial}.

\begin{proof}
Without loss of generality (cf. Corollary~\ref{c:BUNO}) we can assume that $A$ is in its quasidiagonal form.
We show that a suitable scheme with self-similarity $A$ can be obtained as a direct sum of a sufficient number, say K, of Vandermonde schemes corresponding to the polynomial $f$. Suppose that the polynomial $f$ is of degree~$d=r+2t$. Formally, if the polynomial $f$ has $r$ real and $t$ pairs of complex conjugate roots,
$$
f(x)=\prod_{j=1}^{r}(x-\lambda_j)\prod_{k=1}^{t}(x-\mu_k)(x-\overline{\mu_k})\,,\quad r+t \geq 2\,.
$$
Then $M=\max\{l_j:j=1,\dots,r+t\}$, where $l_j=m_A(\lambda_j)$ for $j=1,\dots,r$, $l_{j+k}=m_A(\mu_k)$ for $k=1,\dots,t$ are multiplicities of $\lambda_j\in\R$, $\mu_k\in\C\setminus\R$ in the spectrum of $A$. Set $u:=r+t$ and $K=\max\Bigl\{ M, \Bigl\lceil\tfrac{1}{u-1}\sum_{j=1}^{u}l_j\Bigr\rceil\Bigr\}$.
%
%

Consider $K$ Vandermonde schemes $\Lambda_1$, \dots, $\Lambda_K$, having self-similarities $A_1$, \dots, $A_K$, where $A_i$ are suitable matrices in quasidiagonal form. Each scheme $\Lambda_i$ is given by a lattice $\L_i$ with an associated matrix $L_i$ satisfying
$\Big(\begin{smallmatrix}
 A_i & O \\
 O & B_i
\end{smallmatrix}\Big)L_i = L_iC_f$, where $C_f$ is the companion matrix to $f$. Moreover, the matrices $A_i$ are chosen in such a way that $A\sim A_1\oplus \cdots\oplus A_K$ and to each $A_j$, the corresponding matrix $B_j$ is non-empty. This algebraic situation is solved by the combinatorial task given in Lemma~\ref{l:K}. Consider a family of $K$ identical lists, whose items are
$$
\alpha_1, \alpha_2, \cdots, \alpha_r, \alpha_{r+1}, \cdots, \alpha_{r+t},
$$
where $\alpha_i$ are real roots or pairs of complex conjugate roots of the polynomial $f$, formally
$$
\alpha_i = \begin{cases}
              \lambda_j, & \mbox{for } j\in \{1,2,\cdots , r \}, \\
              \{\mu_{j-r}, \overline{\mu_{j-r}} \}, &  \mbox{for } j\in \{r+1,\dots,r+t \}.
            \end{cases}
$$
Consider a well distributing matrix $\K\in\{0,1\}^{u\times K}$ whose existence is ensured by Lemma~\ref{l:K}. If $\K_{ij}=1$, put $\alpha_i$ into the spectrum of $A_j$, otherwise, put $\alpha_i$ into the spectrum of $B_j$. The direct sum of Vandermonde schemes constructed in this way has for self-similarity a matrix $A_1\oplus\cdots\oplus A_K$. The scheme with self-similarity $A$ is obtained by a simple permutation.
\end{proof}

\begin{pozn}
Note that a direct sum of Vandermonde schemes of the same polynomial $f$ is determined using a lattice $\L$ whose vectors
have coordinates in the splitting field $\F$ of the polynomial $f$. In the following section we provide a construction of a scheme which has smaller lattice dimension compared to that of Proposition~\ref{p:notmindimension}. However, reducing the lattice dimension is achieved at the expense of having coordinates in a field with higher dimension over $\Q$.
\end{pozn}


\section{Minimal dimension of a scheme with given self-similarity}\label{sec:mindimenionscheme}

The construction of a cut-and-project scheme with a given self-similarity $A$ given in the previous section was a simple one. The main advantage was that the resulting scheme is a direct sum of elementary schemes, and the corresponding lattice has generators with components in the algebraic extension given by eigenvalues of $A$.
However, the dimension of the constructed scheme is not a~minimal one. In this section we determine a better bound on the minimal dimension and show that it is actually achieved.

Note that based on Proposition~\ref{p:Anekonjug}, it suffices to examine the cases where the eigenvalues of the matrix $A$ all have the same minimal polynomial $f$. More precisely, suppose that $f\in\Q[X]$ is an irreducible polynomial with roots $\beta_1,\beta_2,\dots,\beta_d$. We will assume that $\sigma(A)\subset\{\beta_1,\dots,\beta_d\}$.

\begin{thm}\label{t:mindimension}
Let $A\in\R^{n\times n}$ be a matrix diagonalizable over $\C$. Suppose that
all $\lambda\in\sigma(A)$ are algebraic integers with the same minimal polynomial $f$ of degree $d$.
Denote
$$
M:=\max\{m_A(\beta):\beta\text{ root of }f\}\geq 1,\qquad m:=\min\{m_A(\beta):\beta\text{ root of }f\}\geq 0.
$$
Let $\Lambda=(\L\subset\R^s,\R^n)$ be a generic cut-and-project scheme with self-similarity $A$ and such that the dimension $s$ of the lattice is as small as possible. Then
$$
s=\begin{cases}
        Md & \mbox{if } m<M, \\
        (M+1)d, & \mbox{otherwise}.
      \end{cases}
$$
\end{thm}

Before presenting the proof, we demonstrate the idea on the self-similarity $A$ taken from Example~\ref{ex:dimenze6}.

\begin{example}\label{ex:dimenze4}
Let $A$, $B$, $C_f$, $C$, $Y$ be as in Example~\ref{ex:dimenze6}.
Note that the matrix $Y$ can be written as
$$
    Y=\begin{pmatrix}
        1 & 1 & 0 & 0 & 0 & 0 \\
        \tau & \tau' & 0 & 0 & 0 & 0 \\
        0 & 0 & 1 & 1 & 0 & 0  \\
        0 & 0 & \tau & \tau' &  0 & 0 \\
        0 & 0 & 0 & 0 & 1 & 1 \\
        0 & 0 & 0 & 0 & \tau & \tau'
      \end{pmatrix}
      \underbrace{
      \begin{pmatrix}
        1 & 0 & 0 & 0 & 0 & 0 \\
        0 & 0 & 0 & 1 & 0 & 0 \\
        0 & 1 & 0 & 0 & 0 & 0  \\
        0 & 0 & 0 & 0 & 1 & 0 \\
        0 & 0 & 1 & 0 & 0 & 0 \\
        0 & 0 & 0 & 0 & 0 & 1
      \end{pmatrix}}_{P}=\left(I_3\otimes \begin{pmatrix}
                                           1 & 1 \\
                                           \tau & \tau'
                                         \end{pmatrix}\right)P\,,
$$
where $\otimes$ stands for the Kronecker product.
Using the same matrix $P$, we can transform
$$
\begin{pmatrix}
A & O \\
O & B
\end{pmatrix} = P^\top \begin{pmatrix}
                      \tau &  &  &  &  &  \\
                       & \tau' &  &  &  &  \\
                       &  & \tau &  &  &  \\
                       &  &  & \tau' &  &  \\
                       &  &  &  & \tau &  \\
                       &  &  &  &  & \tau'
                    \end{pmatrix}P=P^\top \left(I_3\otimes
\begin{pmatrix}
  \tau & 0 \\
  0 & \tau'
\end{pmatrix}\right) P .
$$
Obviously also $C=I_3\otimes C_f$. Substituting these into
$CY=Y\left(\begin{smallmatrix}
  A & O \\
  O & B
\end{smallmatrix}\right)$,
we obtain
$$
(I_3\otimes C_f)\left(I_3\otimes \begin{pmatrix}
                                           1 & 1 \\
                                           \tau & \tau'
                                         \end{pmatrix}\right)P
= \left(I_3\otimes \begin{pmatrix}
                                           1 & 1 \\
                                           \tau & \tau'
                                         \end{pmatrix}\right)P
P^\top \left(I_3\otimes
\begin{pmatrix}
  \tau & 0 \\
  0 & \tau'
\end{pmatrix}\right) P\,,
$$
which gives
$$
(I_3\otimes C_f)\left(I_3\otimes \begin{pmatrix}
                                           1 & 1 \\
                                           \tau & \tau'
                                         \end{pmatrix}\right)
= \left(I_3\otimes \begin{pmatrix}
                                           1 & 1 \\
                                           \tau & \tau'
                                         \end{pmatrix}\right)
\left(I_3\otimes
\begin{pmatrix}
  \tau & 0 \\
  0 & \tau'
\end{pmatrix}\right)\,,
$$
reducing to the obvious
\begin{equation}\label{eq:obviouselement}
  C_f\begin{pmatrix}
   1 & 1 \\
   \tau & \tau'
   \end{pmatrix} =
   \begin{pmatrix}
   1 & 1 \\
   \tau & \tau'
   \end{pmatrix}
\begin{pmatrix}
  \tau & 0 \\
  0 & \tau'
\end{pmatrix}\,.
\end{equation}

The lattice $\L$ in the composition of the elementary cut-and-project schemes into their direct product is defined by the inverse $L=Y^{-1}$ of the matrix $Y$ given using Kronecker 
product with the identity matrix.
One can ask whether the Kronecker product with another non-singular matrix, say $H$, could not produce better results, in particular with respect to the number of elementary schemes needed. Indeed, it turns out that in many cases, suitable choice of the matrix $H$ does reduce the dimension of the resulting cut-and-project scheme.

For the matrix $A$ of Example~\ref{ex:dimenze6}, we will try to compose only two elementary schemes, but instead of $I_2$, we put a general non-singular matrix $H\in\R^{2\times 2}$. We set
$$
Y=\left(H\otimes \begin{pmatrix}
                                           1 & 1 \\
                                           \tau & \tau'
                                         \end{pmatrix}\right)P,
$$ 
as
$$
(I_2\otimes C_f)\left(H\otimes \begin{pmatrix}
                                           1 & 1 \\
                                           \tau & \tau'
                                         \end{pmatrix}\right)
= \left(H\otimes \begin{pmatrix}
                                           1 & 1 \\
                                           \tau & \tau'
                                         \end{pmatrix}\right)
\left(I_2\otimes
\begin{pmatrix}
  \tau & 0 \\
  0 & \tau'
\end{pmatrix}\right)\,.
$$
For $H=I_2$, it was explained in Example~\ref{ex:dimenze6} that the constructed scheme is not generic. Taking a different $H$ may produce a generic scheme.
\end{example}

The following two lemmas will be used in the proof of Theorem~\ref{t:mindimension} for verifying that the constructed scheme is generic. We will consider a matrix $H$ in the form
\begin{equation}\label{eq:H}
H=\begin{pmatrix}
    1 & -t_1 & 0 & \dots & 0 & 0\\
    0 & 1 & -t_2 & \dots & 0 & 0\\
    \vdots & \vdots  & \vdots & \ddots & \vdots  & \vdots \\
    0 & 0 & 0 & \dots & 1 & -t_{K-1} \\
    0 & 0 & 0 & \dots & 0 & 1
  \end{pmatrix}\in\R^{K\times K}\,.
\end{equation}
Recall that $Z\in\C^{d\times d}$ is the Vandermonde matrix whose columns are vectors $\bs{z}_j=(1,\beta_j,\dots,\beta_j^{d-1})^\top$, $j=1,\dots,d$, $\beta_j$ being the roots of the polynomial $f$.

\begin{lem}\label{l:H}
There exist $t_1,\dots,t_{K-1}\in\R$ such that if $\bs{x}\in \C^{Kd}$ of the form $\bs{x}^\top=(x_1,x_2,\dots,x_{Kd})$ satisfies $\prod_{j=1}^dx_j=0$ and $(H\otimes Z)\bs{x}\in\Q^{Kd}$
then $\bs{x}=\bs{0}$.
\end{lem}

\begin{proof}
Note that the condition $\prod_{j=1}^dx_j=0$ says that at least one of the first $d$ components of the vector $\bs{x}$ is zero. We will make the proof for $\bs{x}$ of the form
$\bs{x}^\top=(0,x_2,\dots,x_{Kd})$.
Denote $\bs{r}=(H\otimes Z)\bs{x}\in\Q^{Kd}$. Then
\begin{equation}\label{eq:Hinverse}
\bs{x}=(H^{-1}\otimes Z^{-1})\bs{r}\,.
\end{equation}
We will only use first $d$ rows of the above equation. For simplicity, denote 
$\bs{r}^\top=
  (\bs{r}_1^\top,\bs{r}_2^\top,\dots,\bs{r}_{K}^\top)$, for  $\bs{r}_j\in\Q^{d}\,.$
Calculating from~\eqref{eq:H} the first row of the matrix $H^{-1}$ as
$$
(H^{-1})_{1\bullet}=(1,\,t_1,\,t_1t_2,\,\dots,\,t_1t_2\cdots t_{K-1}),
$$
one then derives from~\eqref{eq:Hinverse}
that
\begin{equation}\label{eq:firstrow}
\begin{pmatrix}
  0 \\
  x_2 \\
  \vdots \\
  x_{d}
\end{pmatrix}=
  Z^{-1}\bs{r}_1 + t_1Z^{-1}\bs{r}_2 +  t_1t_2Z^{-1}\bs{r}_3 + \dots + t_1t_2\cdots t_{K-1}Z^{-1}\bs{r}_K\,.
\end{equation}
For $j=1,\dots,K$, the components of vectors $Z^{-1}\bs{r}_j$ belong to the splitting field $\F=\Q(\beta_1,\dots,\beta_d)$ of the polynomial $f$. Obviously,
one can chose $t_1,\dots,t_{K-1}$ such that the coefficients $1,t_1,t_1t_2,\dots,t_1\cdots t_{K-1}$ are linearly independent over $\F$. Inspecting the first row of~\eqref{eq:firstrow}, we derive that the first component of the vector $Z^{-1}\bs{r}_j$ for $j=1,\dots,k$ vanish. By Corollary~\ref{c:vektor0} we have that
$\bs{r}_1=\cdots =\bs{r}_K=\bs{0}$ and from~\eqref{eq:Hinverse} we get that $\bs{x}=\bs{0}$.
\end{proof}

%

\begin{lem}\label{l:H2}
Denote $\bs{e}_j$, $j=1,\dots,d$, the standard basis in $\R^d$.
Then there exist $\alpha_1,\dots,\alpha_{K}\in\C$ such that for every $\bs{q}\in\Q^{Kd}$, $\bs{q}\neq \bs{0}$, and every vector $\bs{v}\in\C^{Kd}$ of the form $\bs{v}^\top=(\alpha_1\bs{e}_{p_1}^\top,\dots,\alpha_K\bs{e}_{p_K}^\top)$ with $p_1,\dots,p_K\in\{1,\dots,d\}$ it holds that $\bs{q}^\top(H\otimes Z)\bs{v}\neq 0$, where $H$ is from~\eqref{eq:H} with $t_1,\dots,t_{K-1}$ found in Lemma~\ref{l:H}.
\end{lem}

\begin{proof}
  Write $\bs{q}^\top=(\bs{q}_1^\top,\dots,\bs{q}_K^\top)$, $\bs{q}_j\in\R^d$. Then
  $$
  \bs{q}^\top(H\otimes Z)\bs{v} = \alpha_1\bs{q}_1^\top Z\bs{e}_{p_1} + \alpha_2(-t_1\bs{q}_1^\top+\bs{q}_2^\top)Z\bs{e}_{p_2} + \cdots + \alpha_K(-t_{K-1}\bs{q}_{K-1}^\top+\bs{q}_K^\top)Z\bs{e}_{p_K}.
  $$
  We know that $\bs{q}_i^\top Z\bs{e}_j\in\Q(\beta_j)$  for every $i\in\{1,\dots,K\}$, $j\in\{1,\dots,d\}$. Then $\bs{q}^\top(H\otimes Z)\bs{v}$ can be written in the form
  $$
  \bs{q}^\top(H\otimes Z)\bs{v}=\alpha_1c_1+\alpha_2c_2+\cdots + \alpha_Kc_K
  $$
  where we have denoted
  $$
  c_1=\bs{q}_1^\top Z\bs{e}_{p_1},\quad\text{ and }\quad c_j=(-t_{j-1}\bs{q}_{j-1}^\top+\bs{q}_j^\top)Z\bs{e}_{p_j},\quad\text{ for }j=2,\dots,K.
  $$
  Note that $c_j\in{\mathbb G}:=\Q(\beta_1,\dots,\beta_d,t_1,\dots,t_{K-1})$. Choose $\alpha_1,\dots,\alpha_K$ so that they are linearly independent over the field ${\mathbb G}$. Then
equality $\bs{q}^\top(H\otimes Z)\bs{v}=0$ implies that $c_1=c_2=\cdots=c_K=0$. Since components of the vector $(Z\bs{e}_j)^\top=\bs{z}_j^\top =(1,\beta_j,\dots,\beta_j^{d-1})^\top$
are linearly independent over $\Q$, equality $\bs{q}_i^\top Z\bs{e}_j=0$ implies $\bs{q}_i=\bs{0}$. Therefore
$$
\begin{array}{lcl}
c_1=\bs{q}_1^\top Z\bs{e}_{p_1}=0 &\implies &\bs{q}_1=\bs{0},\\[2mm]
c_2=(-t_1\bs{q}_1^\top+\bs{q}_2^\top)Z\bs{e}_{p_2} = \bs{q}_2^\top Z\bs{e}_{p_2}=0 &\implies &\bs{q}_2=\bs{0}, \\[2mm]
\text{etc.}
\end{array}
$$
Thus $\bs{q}^\top(H\otimes Z)\bs{v}=0$ implies that $\bs{q}=\bs{0}$ as we wanted to show.
\end{proof}

\begin{proof}[Proof of Theorem~\ref{t:mindimension}]
The case when the polynomial $f$ is linear or quadratic with negative discriminant is solved in Section~\ref{sec:trivial}, therefore we can assume for the rest of the proof that $f$ is not of such form.

Set $K:=M$ if $m<M$ and $K:=M+1$ otherwise. This choice of $K$ ensures that there exist two distinct roots of the polynomial $f$ say $\beta_1\neq \overline{\beta_2}$ such that
\begin{equation}\label{eq:111}
  m_A(\beta_1)\geq 1\quad\text{ and }\quad 0\leq m_A(\beta_2)<K.
\end{equation}
We will show two statements:
\begin{itemize}
\item[(a)] If $A$ is a self-similarity of a generic cut-and-project scheme $(\L\subset\R^s,\R^n)$, then $s\geq Kd$.
\item[(b)] There exists a generic cut-and-project scheme $(\L\subset\R^s,\R^n)$ with self-similarity $A$ such that $s= Kd$.
\end{itemize}

In order to prove (a), consider a generic cut-and-project scheme $(\L\subset\R^s,\R^n)$ with self-similarity $A$. By Definition~\ref{de:self} there exist unique matrices $B\in\R^{(s-n)\times(s-n)}$, $C\in\Z^{s\times s}$ satisfying~\eqref{eq:ABC}. Take a root $\lambda$ of the polynomial $f$ such that $m_A(\lambda)=M$. Then the characteristic polynomial $\chi_A$ of $A$ is divisible over $\C$ by $(X-\lambda)^M$. Since $\chi_C=\chi_A\chi_B\in\Z[X]$, the polynomial $\chi_C$ must be divisible by $f^M$. Therefore for the degree $s$ of $\chi_C$ we can write $s\geq Md$.
If $K=M$, i.e.\ $m<M$, then (a) is established. Suppose that $K=M+1$, i.e.\ $m=M$. This means that $m_A(\lambda)=M$ for every root $\lambda$ of $f$. By Corollary~\ref{c:spektrumAB} there exists a root $\lambda'$ of $f$ such that $m_B(\lambda')\geq 1$. Therefore $m_C(\lambda')=m_A(\lambda')+m_B(\lambda')\geq M+1$. Consequently, $\chi_C$ is divisible by $f^{M+1}$, and its degree $s$ thus satisfies $s\geq Kd$.

Item (b) will be proved by providing a construction for $s=Kd$ of a matrix $C\in\Z^{s\times s}$, $B\in\R^{(s-n)\times(s-n)}$ and $L\in\R^{s\times s}$ such that~\eqref{eq:ABC} holds and such that the lattice $\L$ associated to $L$ gives a generic cut-and-project scheme $(\L\subset\R^s,\R^n)$. Then by Proposition~\ref{p:ABCzpet}, $A$ is a self-similarity of the scheme. The matrix $L$ is given as $L=Y^{-1}$ where $Y\in\R^{s\times s}$ is a suitable matrix satisfying~\eqref{eq:ABC-Y}.
Without loss of generality, let $A$ be given in its real Jordan form, cf. Corollary~\ref{c:BUNO}.

Denote $C_f$ the companion matrix of the polynomial $f$.
Set $D_f$ to be the real Jordan form of $C_f$.
Further denote $Y_f$ the matrix whose columns are vectors $(1,\beta_j,\dots,\beta_j^{d-1})^\top\in\R^d$, or $(1,\Re\beta_j,\dots,\Re\beta_j^{d-1})^\top$,
$(1,\Im\beta_j,\dots,\Im\beta_j^{d-1})^\top$ in cases that $\beta_j$ is non-real.
We have the equality $C_fY_f=Y_fD_f$.

The choice $s=Kd>n=\sum_{\lambda\in\sigma(A)} m_A(\lambda)$ ensures that there exists a square matrix $B$ of order $s-n\geq 1$ such that
\begin{equation}\label{eq:222}
  P^T(I_K\otimes D_f) P =
  \begin{pmatrix}
  A & O \\
  O & B
  \end{pmatrix}
\end{equation}
for a permutation matrix $P\in\{0,1\}^{s\times s}$. Such a matrix $P$ is not given uniquely. We choose $P$ in such a way that at least one eigenvalue of the first block $D_f$ of the matrix $I_K\otimes D_f$ contributes to the spectrum of $A$ and at least one eigenvalue of the first block $D_f$ contributes to the spectrum of $B$. Possibility of such a choice is ensured by~\eqref{eq:111}. Consequently, at least one eigenvalue from each block $D_f$ of the matrix $I_K\otimes D_f$ contributes to the spectrum of $A$.

We will show that the desired matrices $C, Y$ can be found in the form
$$
C=I_K\otimes C_f \quad\text{ and }\quad Y=(H\otimes Y_f)P
$$
where $P$ is the above permutation matrix and $H\in\R^{K\times K}$ is a suitable matrix, whose choice will be described later.
First we verify validity of~\eqref{eq:ABC-Y},
$$
\begin{aligned}
CY &=(I_K\otimes C_f)(H\otimes Y_f)P = (H\otimes Y_f)(I_K\otimes D_f)P=\\
&=(H\otimes Y_f)(I_K\otimes D_f)P = YP^\top(I_K\otimes D_f)P= Y\begin{pmatrix}
A & O \\
O & B
\end{pmatrix}\,.
\end{aligned}
$$

It remains to check that there exists a suitable choice of $H\in\R^{K\times K}$ such that the lattice $\L$ associated to the matrix $L=Y^{-1}$ defines a generic cut-and-project scheme. We choose $H$ in the form~\eqref{eq:H} which satisfies the assumptions of Lemma~\ref{l:H}. Realize that
\begin{equation}\label{eq:HtenzorY}
H\otimes Y_f = \begin{pmatrix}
    Y_f & -t_1Y_f & O & \dots & O & O\\
    O & Y_f & -t_2Y_f & \dots & O & O\\
    \vdots & \vdots  & \vdots & \ddots & \vdots  & \vdots \\
    O & O & O & \dots & Y_f & -t_{K-1}Y_f \\
    O & O & O & \dots & O & Y_f
  \end{pmatrix}\,.
\end{equation}
The permutation matrix $P$ is chosen in such a way that $Y=(H\otimes Y_f)P$ written in the form $Y=(Y_1,Y_2)$, $Y_1\in\R^{s\times n}$, $Y_2\in\R^{s\times(s-n)}$ has the following property.
At least one of the first $d$ columns of the matrix $H\otimes Y_f$ belongs to $Y_1$ and at least one of the first $d$ columns of $H\otimes Y_f$ belongs to $Y_2$. Moreover, at least one column, with index say $p_i\in\{1,\dots,d\}$, of each block of the form
$$
\begin{pmatrix}
    O \\[-1mm]
    \vdots\\
    O \\
    -t_iY_f \\
    Y_f \\
    O \\[-1mm]
    \vdots\\
    O
  \end{pmatrix} \in\R^{Kd\times d}
$$
belongs to $Y_1$.
Non-degeneracy, resp.\ aperiodicity is valid, if no non-zero rational vector can be obtained by complex linear combination from the columns of $Y_2$, resp. from the columns of $Y_1$, see Proposition~\ref{p:generic}. This is ensured by Lemma~\ref{l:H}, realizing that any complex combination of columns of the matrix $H\otimes Y_f$ can be obtained by a complex combination of columns of $H\otimes Z$, where $Z$ is the Vandermonde matrix.

Last, we verify that the cut-and-project scheme is irreducible. For that we will verify validity of Item (iii) of Proposition~\ref{p:generic}. For the vector $\bs{y}$ in Proposition~\ref{p:generic} we take $(H\otimes Z)\bs{v}$ from Lemma~\ref{l:H2}, where $p_1,\dots,p_d$ are defined above by the permutation $P$.
\end{proof}

\begin{pozn}
  Compare the dimension of the cut-and-project schemes constructed in Proposition~\ref{p:notmindimension} and Theorem~\ref{t:mindimension}. In many cases, these dimensions coincide, as is illustrated on the self-similarity in Example~\ref{ex:kubicke}. The dimension of the lattice is 6, and it is minimal, as $M=2$ and $d=3$.
\end{pozn}

Sections~\ref{sec:answerQ2} and~\ref{sec:mindimenionscheme} give a complete answer to Question~\ref{q2} for a mapping $A$ diagonalizable over $\C$. In particular, we have described such mappings $A$, for which a cut-and-project scheme with self-similarity $A$ exists and we explained how to obtain a scheme with lattice dimension $s$ minimal possible.
We therefore have a necessary condition which a mapping $A$ should satisfy if it is a self-similarity of a cut-and-project set, see Proposition~\ref{p:q1->q2}.

\section{Answer to Question~\ref{q1}}

Let us come back to the original Question~\ref{q1}, namely of the existence of a cut-and-project set $\Sigma(\Omega)\subset\R^n$ with a given diagonalizable self-similarity $A\in\R^{n\times n}$.
The aim of this section is the prove the necessary and sufficient condition for a mapping $A$ so that there exists a cut-and-project set $\Sigma(\Omega)$ which has $A$ as a self-similarity. Recall that in our setting, a cut-and-project set is defined in a generic cut-and-project scheme, i.e. it has no translational symmetry.

\begin{thm}\label{thm:npp_for_cps}
Let $A\in\R^{n\times n}$ be a non-singular matrix diagonalizable over $\C$. Then there exists a cut-and-project set $\Sigma(\Omega)$ satisfying $A\Sigma(\Omega)\subset\Sigma(\Omega)$ if and only if the spectrum of $A$ has the following properties $\mathfrak{P}$:
\begin{itemize}
\item[$(\mathfrak{P}1)$] Every eigenvalue of $A$ is an algebraic integer.
\item[$(\mathfrak{P}2)$] Every complex number $\mu$ of modulus $|\mu|>1$, algebraically conjugate to an eigenvalue of $A$, is an eigenvalue of $A$ as well. Moreover,
for all $\nu \al \mu$ we have $m_A(\mu) \geq m_A(\nu)$ and the inequality is strict for at least one $\nu \al \mu$.
\end{itemize}
\end{thm}

The proof of Theorem~\ref{thm:npp_for_cps} follows. Subsection~\ref{s:thmpomoc} contains several simple auxiliary facts. In Subsection~\ref{s:thmnutna} we demonstrate necessity of $\mathfrak{P}$. Sufficiency of $\mathfrak{P}$ is shown in Subsection~\ref{s:thmpostacujici}. The last subsection explains the consequences of Theorem~\ref{thm:npp_for_cps} and provides a formula for the minimal dimension of the cut-and-project scheme allowing a cut-and-project set with self-similarity $A$.

\subsection{Auxiliary facts}\label{s:thmpomoc}

Suppose one has a generic scheme $\Lambda=(\L\subset\R^s,\R^n)$ with self-similarity $A$ given by a non-singular matrix $A \in \R^{n \times n}$ diagonalizable over $\C$.
If such a setting is established we say that \textit{basic assumptions} are fulfilled.
Denote as usual by $L\in \R^{s\times s}$ a matrix associated to the lattice $\L$. Denote by $B \in \R^{(s-n) \times (s-n)}$ and $C\in\Z^{s\times s}$ the matrices from  Definition~\ref{de:self}. Note that these matrices are defined uniquely.

\begin{prop}
	\label{p:vlastnostiCPS}
	Let the basic assumptions be satisfied and let $\Omega\subset\R^{s-n}$ be a bounded set satisfying $\overline{\Om^\circ} = \overline{\Om}$. Then the following claims hold:
	\begin{itemize}
		\item[(i)] If $B\Om \subset \Om$ then $A\S(\Om) \subset \S(\Om)$.
		\item[(ii)] If $A\S(\Om) \subset \S(\Om)$ then  $B\overline{\Om} \subset \overline{\Om}$.
	\end{itemize}
\end{prop}

\begin{proof}
Let us first realize that since $A$ is a non-singular matrix, then so is $B$ (cf.\ Proposition~\ref{p:minpol}). Using the matrix formalism we have
\begin{equation}\label{eq:prevodycapmnozin}
\begin{aligned}
A\S(\Om) &= A\left\{\pi_\parallel(\bs{x}): \bs{x}\in \L,\ \pi_\perp(\bs{x})\in \Om \right\} =\\
    &=  \left\{A(I_n \ O)L\bs{r} : \bs{r}\in \Z^s, \ B(O \ I_{s-n})L\bs{r} \in B\Om \right\} =\\
		&= \left\{ (I_n \ O ) LC\bs{r}: \bs{r}\in \Z^s, \ (O \ I_{s-n})LC\bs{r} \in B\Om \right\}.
\end{aligned}
\end{equation}

In order to prove Item (i), we further continue to see that
$$
A\S(\Om)\subset \left\{ (I_n \ O ) L\bs{q}: \bs{q}\in \Z^s, \ (O \ I_{s-n})L\bs{q} \in B\Om \right\}= \S(B\Om) \subset \S(\Om).
$$

For the proof of Item (ii), recall that the cut-and-project set depends on the given lattice. We will consider the original lattice $\L$ associated to a matrix $L$, and a lattice $\widetilde{\L}$ associated to the matrix $LC$. Since $C$ is also a non-singular matrix, the lattice $\tilde{\L}$ is of full dimension $s$ and according to Lemma~\ref{l:BUNO} the cut-and-project scheme $(\tilde{\L}\subset\R^s,\R^n)$ is generic. We derive from~\eqref{eq:prevodycapmnozin} that
\begin{equation}\label{eq:dk4}
		A\S_\L(\Om) = \S_{\widetilde{\L}} \left(B\Om\right) \subset \pi_\parallel(\widetilde{\L}).
\end{equation}
Recall the mappings $\ast:\pi_{\parallel}(\L)\to\pi_{\perp}(\L)$, $\tilde{\ast}:\pi_{\parallel}(\tilde{\L})\to\pi_{\perp}(\tilde{\L})$. Since
$\widetilde{\L} \subset \L$, the operation $\widetilde{\ast}$ can be seen as a restriction of $\ast$ on $\pi_\parallel(\widetilde{\L})$, both mappings defined by
$\pi_\perp\circ\pi_\parallel^{-1}$. We have
\[
 \left(A\S_\L(\Om) \right)^{\ast} = \left(A\S_\L(\Om) \right)^{\widetilde{\ast}} =  \left(\S_{\widetilde{\L}} \left(B\Om\right) \right)^{\widetilde{\ast}} = \pi_\perp(\widetilde{\L}) \cap B\Om,
 \]
 where we have used~\eqref{eq:dk4} and~\eqref{eq:dk1}. From~\eqref{eq:dk2} we then derive
 		\begin{equation}
		\label{eq:dk5}
		\overline{\left(A \S_\L(\Om)  \right)^\ast} = \overline{B\Om}=B\overline{\Om}.
		\end{equation}
On the other hand the assumption $A\S_\L(\Om) \subset \S_\L(\Om)$ gives
		\[
\overline{\left(A\S_\L(\Om)\right)^\ast} \subset \overline{\left(\S_\L(\Om)\right)^\ast} = \overline{\Om}
\]
where the relation \eqref{eq:dk2} was used. Combining this result together with \eqref{eq:dk5} reads
$B\overline{\Om} \subset \overline{\Om}$, as required.
\end{proof}

The following claim can be easily shown.

\begin{claim}\label{claimB}
Let $B\subset\R^{(s-n)\times(s-n)}$ be a non-singular matrix diagonalizable over $\C$. Then there exists a bounded  $\Om \subset \R^{s-n}$ such that $\overline{\Om^\circ} = \overline{\Om}\neq\emptyset$ and $B\overline{\Om} \subset \overline{\Om}$ if and only if the spectral radius $\varrho(B)$ of the matrix $B$ is smaller than or equal to one.
\end{claim}

\subsection{Necessary condition}\label{s:thmnutna}

As an immediate consequence of Item (ii) of Proposition~\ref{p:vlastnostiCPS} and Claim~\ref{claimB} one derives the following property of the matrix $B$, provided we have a cut-and-project set $\Sigma(\Om)$ with self-similarity $A$.

\begin{coro}
	\label{c:dusl5}
	Let the cut-and-project scheme $\Lambda=(\L\subset\R^s,\R^n)$ and the matrix $A\in\R^{n\times n}$ satisfy basic assumptions. If $A\S(\Om) \subset \S(\Om)$ for a bounded $\Om\subset \R^{s-n}$  with $\overline{\Om^\circ} = \overline{\Om}\neq\emptyset$, then the spectral radius $\varrho(B)$ of the matrix $B$ corresponding to $A$ by Definition~\ref{de:self} is smaller than or equal to 1.
\end{coro}

The following proposition states one of the implications of Theorem~\ref{thm:npp_for_cps}.

\begin{prop}
	\label{thm:vetaT}
	Let $A \in \R^{n \times n}$ be a non-singular matrix diagonalizable over $\C$. Suppose there exists a generic cut-and-project scheme $\Lambda = (\L \subset \R^s, \R^n)$ such that $A$ is a self-similarity of a cut-and-project set $\S(\Om)$ derived from $\Lambda$ using a bounded window $\Om \subset \R^{s-n}$ with $\overline{\Om^\circ} = \overline{\Om}\neq\emptyset$.
	Then the spectrum of $A$ satisfies Properties $\mathfrak{P}$.
\end{prop}

\begin{proof}
By Proposition~\ref{p:q1->q2}, we derive that basic assumptions are satisfied, in particular, $A$ is a self-similarity of the generic cut-and-project scheme $\Lambda$. Thus
validity of $\mathfrak{P}1$ is given directly in Corollary~\ref{c:duseldek1}.

Note that property $\mathfrak{P}2$ needs a proof only in case when an eigenvalue $\lambda$ of $A$ has at least one algebraic conjugate $\mu$ that is in modulus strictly greater than 1.
For further considerations, recall matrices $B \in \R^{(s-n) \times (s-n)}$ and $C \in \Z^{s \times s}$ from Definition~\ref{de:self}.
Suppose that $\lambda\in\sigma(A)$ and $\mu \al \lambda$ and $|\mu| > 1$.
Denote
\begin{equation}\label{eq:oznM}
 M := \max \left\{m_T(\nu) : \nu\al \lambda,\  T \in \{A,B\} \right\} .
\end{equation}
Let $f \in \Z[X]$ be the minimal polynomial of $\lambda$ (and $\mu$) and let $\chi_C$ be the characteristic polynomial of $C$. We know that the minimal polynomial $\m{A}$ of $A$ over $\Q$ must be divisible by $f$. Proposition \ref{p:minpol} says, that $f$ divides $\chi_C$ as well. Denote by $\lambda'$ the algebraic conjugate of $\lambda$ such that $m_A(\lambda') = M$ or $m_B(\lambda') = M$,
i.e.\ $\lambda'$ is the argument of maxima in~\eqref{eq:oznM}.
Using \eqref{eq:ABC} one gets that $m_C(\lambda') \geq M$. Therefore $\chi_C$ must be divisible by $f^M$. Let us show that $\chi_C$ is not divisible by $f^{M+1}$.
Suppose on the contrary that $f^{M+1}$ divides $\chi_C$. Then
	\[
m_A(\mu) + m_B(\mu) = m_C(\mu) \geq M+1.
\]
From the definition of $M$ it follows that $m_A(\mu) \leq M$ and thus $m_B(\mu) \geq 1$ which is a contradiction with Corollary \ref{c:dusl5}. Therefore $m_C(\nu) = M $ for all $\nu \al \lambda$ and equation \eqref{eq:dk7} can be rewritten as
	\begin{equation}
	\label{eq:dk8}
	m_A(\nu) + m_B(\nu) = m_C(\nu)=M  \ \mbox{for all }\nu\al \lambda.
	\end{equation}
Corollary \ref{c:dusl5} states that $m_B(\mu) = 0$ and thus $m_A(\mu) = M$. This in particular means that $\mu\in\sigma(A)$. Moreover, from the definition of $M$ it follows that
\begin{equation}\label{eq:dk7}
 m_A(\nu) \leq M = m_A(\mu)  \ \mbox{for all }\nu \al \lambda.
\end{equation}
If $m_A(\mu) = m_A(\nu)$ for all  $\nu \al \lambda$, then equation \eqref{eq:dk8} implies that $m_B(\nu) = 0$ for all $\nu \al \lambda$ which is a contradiction with Corollary \ref{c:spektrumAB}. Thus the inequality in~\eqref{eq:dk7} is strict for at least one $\nu \al \lambda$.
\end{proof}

\subsection{Sufficient condition}\label{s:thmpostacujici}

It remains to show sufficiency of Properties $\mathfrak{P}$ in Theorem~\ref{thm:vetaT} for the matrix $A\in\R^{n\times n}$ to be a self-similarity of a cut-and-project set.

\begin{prop}
	\label{thm:vetaN}
	Let $A\in \R^{n\times n}$ be a non-singular matrix diagonalizable over $\C$ satisfying Properties $\mathfrak{P}$. Then there exists a generic cut-and-project scheme $\Lambda=(\L\subset\R^s,\R^n)$ and a bounded $\Omega\subset\R^{s-n}$, $\overline{\Omega^\circ}=\overline{\Omega}$ such that $A$ is a self-similarity of the cut-and-project set $\S(\Om)$.
\end{prop}

\begin{proof}
First we show the following claim:

\medskip
{\it There exists a generic cut-and-project scheme $\Lambda = (\L\subset \R^s, \R^n)$ with self-similarity $A$ and the matrix $B$ given from Definition~\ref{de:self} has $\varrho(B) \leq 1$.}
\medskip

Since by $\mathfrak{P}$ the eigenvalues of $A$ are algebraic integers, the equivalence relation $\al$ splits $\sigma(A)$ into a finite number, say $k$, of equivalence classes. Let $W_A \in \R^{n\times n}$ be a matrix such that $W_A^{-1}AW_A$ is a block diagonal matrix with blocks $A_1, \dots, A_k$ with each block $A_i$ having elements of its spectrum mutually conjugated for $i = 1, \dots, k$.  If the above claim holds for each $A_i$ then, according to Lemma \ref{l:sklad} and Corollary~\ref{c:BUNO}, it holds for $A$ as well. In particular, if for each $A_i$ we find a
scheme such that the corresponding matrix $B_i$ has spectral radius $\varrho(B_i)\leq 1$, then also to $A$ we can find a cut-and-project scheme with corresponding matrix $B$ satisfying $\varrho(B)\leq 1$.
	
	Without loss of generality assume that all eigenvalues in $\sigma(A)$ are mutually algebraically conjugate. Let the minimal polynomial $f$ of these eigenvalues be of degree $d$. Theorem \ref{t:mindimension} states that there exists a generic cut-and-project scheme $\Lambda = (\L\subset \R^s,\R^n)$ such that $A$ is its self-similarity. By Definition~\ref{de:self} there exist matrices $B\in\R^{(s-n)\times(s-n)}$, $C\in\Z^{s\times s}$ such that
$\left(\begin{smallmatrix}
	A & O \\ O & B
\end{smallmatrix} \right) L = LC$
where $L$ is a matrix associated to $\L$. Corollary~\ref{c:spektrumAB} implies that each $\nu \in \sigma(B)$ is conjugated to some $\lambda \in \sigma(A)$. Therefore all $\nu \in \sigma(B)$ have the same minimal polynomial $f$. If all roots of $f$ are in modulus smaller than or equal to $1$, then $\varrho(B)\leq 1$ obviously. Suppose there exists a root $\mu$ of $f$ with $|\mu|>1$. Then property $\mathfrak{P}2$ says that all roots of $f$ of modulus strictly greater then 1 have the same multiplicity $M$ and there exists a root $\mu'$ of $f$ such that $m_A(\mu') < m_A(\mu) = M$. Thus the cut-and-project scheme $\Lambda$ whose existence is ensured by Theorem \ref{t:mindimension} is of dimension $s = Md$. This implies for the matrix $C$ that $\chi_C = f^M$ and each element $\nu$ of the spectrum of $C$ has multiplicity $m_C(\nu) = M$. Since
	$m_A(\nu) + m_B(\nu) = m_C(\nu) = M$ and $m_A(\lambda) = M$ for all roots $\lambda$ with $|\lambda| >1$, one obtains the desired result, namely $m_B(\lambda) = 0$ for all $\lambda$ with $|\lambda| >1$ and thus $\varrho(B) \leq 1$. Thus the claim is established.

With this in hand, Claim~\ref{claimB} ensures existence of a bounded window $\Om$ satisfying $\overline{\Om^\circ} = \overline{\Om}\neq\emptyset$ and $B\overline{\Om} \subset \overline{\Om}$. Statement (i) of Proposition \ref{p:vlastnostiCPS} then implies that $A\S(\overline{\Om}) \subset \S(\overline{\Om})$.
\end{proof}

\subsection{Consequences of Theorem~\ref{thm:npp_for_cps}}\label{sec:dusledky}

Theorem~\ref{thm:npp_for_cps} has many immediate consequences.
Suppose that $A\in\R^{n\times n}$ is a non-singular diagonalizable matrix. Then its minimal polynomial $\m{A}$ over $\Q$ can be factorized into a product of distinct monic polynomials irreducible over $\Q$. Let $f\in \Z[X]$ be one of the monic polynomials dividing $\m{A}$.
Suppose that the spectrum of $A$ satisfies properties $\mathfrak{P}$ from Theorem \ref{thm:npp_for_cps}. Then the polynomial $f$ satisfies the following:

\begin{itemize}
\item [(i)] All roots of $f$ that are in modulus strictly greater than 1 have the same multiplicity, say $M$, in the spectrum $\sigma(A)$ of $A$. The other roots have their multiplicities lower than or equal to $M$.
	
\item [(ii)]  If there exists a root of $f$ in modulus strictly greater than 1, then there exists a root of $f$ with modulus strictly lower than 1.
	\subitem \textit{Proof: } Suppose that $\mu$ is a root of $f$ with $|\mu|>1$. If all roots of $f$ are in modulus strictly greater than 1, then, according to the previous Item (i), all roots have the same multiplicity in the spectrum of $A$ and this is a contradiction to property $\mathfrak{P}2$. Therefore at least one root $\lambda$ satisfies $|\lambda|\leq 1$. If there is a root $\lambda$ of $f$ of modulus $|\lambda| = 1$, then $\overline{\lambda} = {\lambda}^{-1}$ is also a root of $f$. This implies that $f$ is a reciprocal polynomial, in particular, $\mu^{-1}$ is also a root of $f$ and $|\mu^{-1}|<1$.

\item [(iii)] If none of the roots of $f$ is in modulus strictly greater than 1, then by Kronecker's theorem~\cite{Kronecker} all roots of $f$ are of modulus 1 and necessarily $f$ is a cyclotomic polynomial. Recall that $d$-th cyclotomic polynomial is defined as the minimal polynomial of a primitive $d$-th root of unity. In this case both $A$ and $B$ are orthogonal matrices.

\item [(iv)] If $f$ is of degree 1, then $f(X) = X\pm 1$, in particular $A$ does not have any eigenvalue $k\in \Z$, $|k|>1$.

\item [(v)] If $f$ is of the degree 2 and its discriminant is negative, then $f(X) = X^2 \pm X + 1$ or $f(X) = X^2+1$, particularly any quadratic number $\lambda \in \C\backslash \R$ different from $\frac{\pm 1 \pm \mathrm{i}\sqrt{3}}{2}$ or $\pm \mathrm{i}$ is not contained in $\sigma(A)$.
	\subitem \textit{Proof: } Let $\lambda \in \sigma(A)$ be a root of $X^2 + pX +q$ with $p^2-4q <0$. Then both $\lambda$ and $\overline{\lambda}$ belong to $\sigma (A)$ and both have the same modulus and multiplicity. Requiring property $\mathfrak{P}2$ yields $|\lambda| = |\overline{\lambda}| = 1$, i.e. $q = |\lambda \overline{\lambda}| = 1$. The condition $p^2-4q <0$ implies $p=0$ or $p= \pm 1$. Note that $X^2 \pm X + 1$, $f(X) = X^2+1$ are the only quadratic cyclotomic polynomials.
	
\end{itemize}

Let us put together the previous ideas and concepts to describe the minimal dimension of a cut-and-project scheme which allows one to construct a cut-and-project set with self-similarity $A\in\R^{n\times n}$ satisfying $\mathfrak{P}$.

\begin{thm}\label{thm:mindimensionCAPset}
	Let $A\in \R^{n\times n}$ be non-singular matrix diagonalizable over $\C$ with Properties $\mathfrak{P}$ from Theorem \ref{thm:npp_for_cps}. Denote by
$f_i\in \Z[X]$ the distinct monic polynomials irreducible over $\Q$ in the factorization
\[
\mu_{A, \Q}(X) = f_1(X) \cdots f_k(X)
\]
and denote by $d_i$ the degree of $f_i$.
Set
\[
M_i := \max \left\{m_A(\lambda) : \lambda \mbox{ is a root of }f_i \right\}.
\]
Define $s_i$ as follows:
	\begin{itemize}
		\item if $f_i$ is not a cyclotomic polynomial set $s_i = M_i d_i$,
		\item if $f_i$ is a cyclotomic polynomial set
$$
s_i = \begin{cases}
		M_i d_i &\mbox{if } M_i > \min \left\{m_A(\lambda) : \lambda \mbox{ is a root of }f_i \right\}, \\
		(M_i +1) d_i & \mbox{otherwise}.
		\end{cases}
$$
	\end{itemize}

Then $s = s_1 + \cdots + s_k$ is the minimal dimension of a lattice that allows one to construct a generic cut-and-project scheme $\Lambda=(\L\subset\R^s,\R^n)$ and a cut-and-project set $\S(\Om)$ with self-similarity $A$.
\end{thm}

\begin{proof}
From Proposition~\ref{p:Anekonjug} and Theorem~\ref{t:mindimension} it follows that the minimal dimension satisfies $s \geq s_1 + \cdots + s_k$. The constructive proof of Proposition~\ref{thm:vetaN} shows that $s = s_1 + \cdots + s_k$ is sufficient.
\end{proof}

%
%
%
%
\section{Comments}\label{sec:comments}

In this paper we have answered Question~\ref{q1} by providing necessary and sufficient conditions on the matrix $A\in\R^{n\times n}$ so as one can find a cut-and-project set without translational symmetry having $A$ as a self-similarity. This allowed us to give a generalization of the results found in~\cite{Lagarias} or~\cite{Aperiodic1} for the special cases where $A$ is a scaling or a rotation. We have described a construction of a suitable cut-and-project scheme which has minimal dimension possible.
Let us compare our results (Theorems~\ref{thm:npp_for_cps} and~\ref{thm:mindimensionCAPset}) to the previous results known for self-similarities $A$ in the form of scalings or rotations

If $A$ is a scaling by a real number $\eta>1$, i.e.\ $A=\eta I_n$, then Theorem~\ref{thm:npp_for_cps} implies that $\eta$ is an algebraic integer and all its algebraic conjugates are in modulus smaller than or equal to $1$. This in particular means that $\eta$ is a Pisot or a Salem number. For the dimension $s$ of any cut-and-project scheme allowing $A=\eta I_n$ as a self-similarity we can derive from Corollary~\ref{c:spektrumAB} that it is divisible by the degree $d$ of the minimal polynomial of $\eta$. Moreover, Theorem~\ref{thm:mindimensionCAPset} states that
the minimal possible dimension is equal to $s=nd$, since $m_A(\eta)=n$. These results correspond to Theorem 4.1 of Lagarias~\cite{Lagarias}. Note that his theorem is valid for a more general class of Delone sets; cut-and-project sets form a subclass.

Suppose now that $A\in\R^{n\times n}$ corresponds to an Euclidean transformation (rotation, reflection, or their combination), i.e.\ $A$ is an orthogonal matrix and its spectrum is contained on the unit circle. Any eigenvalue is a $r$th root of unity for some $r\in\N$. If $r=1,2,3$, or $6$, then it is well known that periodic lattices in $\R^2$ may have rotational symmetry of such order. 
Nevertheless, even a cut-and-project set without any translational symmetry can be invariant under such a rotation.
When considering planar cut-and-project sets with $r$-fold symmetry for $r\notin\{1,2,3,6\}$, then in the spectrum of $A$ we have only $\lambda,\overline{\lambda}=\lambda^{-1}$ where $\lambda$ is a primitive $r$-\text{th} root of unity, the multiplicity $m_A(\lambda)=1$, and thus the dimension $s$ of the scheme is divisible by $d=\phi(r)$, i.e.\ the degree of the $r$-th cyclotomic polynomial. Theorem~\ref{thm:mindimensionCAPset} states that $s=\phi(r)$ is possible. We thus recover the statement of Theorem 3.2 in~\cite{Aperiodic1}.

Solving Question~\ref{q1} we focused on non-singular matrices diagonalizable over $\C$. In this sense, the answer to Question~\ref{q1} is not exhaustive.
One may think of another direction to extend the study of this paper, namely when searching for a construction of a cut-and-project set closed under a set of linear mappings.
Such a construction would generalize the study of Pleasants~\cite{Pleas} who considered a finite group of isometries.

\section{Acknowledgements}

This work was supported by the project CZ.02.1.01/0.0/0.0/16\_019/0000778. We also acknowledge financial support of
the Grant Agency of the Czech Technical University in Prague, grant No.\ SGS17/193/OHK4/3T/14.


\end{document}